\documentclass[12pt]{article}

\usepackage[inline]{enumitem}
\usepackage{float}
\usepackage{mathtools}
\usepackage{graphicx}
\usepackage[round]{natbib}
\usepackage{geometry}
\usepackage{tikz}
\usepackage[english]{babel}
\usepackage{longtable}
\usepackage{color}
\usepackage{amssymb,amsmath,amsthm}
\usepackage{multirow}
\usepackage[titletoc,title]{appendix}
\usepackage{authblk}
\usepackage{setspace}
\usepackage{dsfont}
\usepackage[OT1]{fontenc}
\usepackage{subcaption}
\usepackage{refcount}

\usepackage{xr-hyper}
\usepackage[colorlinks,citecolor=blue,urlcolor=blue]{hyperref}

\usepackage[inline]{trackchanges}
\addeditor{Ivan}
\addeditor{Nima}

\newtheorem{theorem}{Theorem}

\newtheorem{remark}{Remark}
{
  \theoremstyle{definition}
  
}
{
  \theoremstyle{definition}
  \newtheorem{assumptioniden}{}
}
{
  \theoremstyle{definition}
  \newtheorem{example}{Example}[section]
}

\newcommand{\supp}{\mathop{\mathrm{supp}}}

\newtheorem{lemma}{Lemma} \newtheorem{coro}{Corollary}

\DeclareMathOperator{\expit}{expit}
\DeclareMathOperator{\bern}{Bern}
 \DeclareMathOperator{\var}{Var}

\newcommand{\indep}{\mbox{$\perp\!\!\!\perp$}} 
 \newcommand{\dd}{\mathrm{d}}
\newcommand{\Pn}{\mathbb{P}_{n}}
\newcommand{\Gn}{\mathbb{G}_{n}}
\newcommand{\Gnj}{\mathbb{G}_{n,j}}
\newcommand{\Pnj}{\mathbb{P}_{n,j}}

\newcommand{\thetasub}{\hat\theta_{\mbox{\scriptsize sub}}(\delta)}
\newcommand{\thetare}{\hat\theta_{\mbox{\scriptsize re}}(\delta)}

\newcommand{\thetaaipw}{\hat\theta(\delta)}

\newcommand{\hgd}{\hat g_\delta}
\newcommand{\hm}{\hat m}
\newcommand{\prob}{\mathbb{P}}
\newcommand{\one}{\mathds{1}} \renewcommand{\P}{\mathbb{P}}
 
\renewenvironment{proof}{{\it Proof }}{\qed \\}

\renewcommand{\P}{\mathbb{P}}

\newcommand{\E}{\mathbb{E}}
\newcommand{\M}{\mathcal{M}}
\DeclarePairedDelimiterX{\norm}[1]{\lVert}{\rVert}{#1}

\pgfdeclarelayer{background}
\pgfsetlayers{background,main}
\usetikzlibrary{arrows,positioning}
\tikzset{
>=stealth',
punkt/.style={
rectangle,
rounded corners,
draw=black, very thick,
text width=6.5em,
minimum height=2em,
text centered},
pil/.style={
->,
thick,
shorten <=2pt,
shorten >=2pt,}
}

\newcommand{\Vertex}[2]
{\node[minimum width=0.6cm,inner sep=0.05cm] (#2) at (#1) {$\footnotesize#2$};
}
\newcommand{\Vertexr}[2]
{\node[rectangle, draw, minimum width=0.6cm,inner sep=0.05cm] (#2) at (#1) {$\footnotesize#2$};
}
\newcommand{\EdgeR}[3]%
{ \begin{pgfonlayer}{background}
\draw[dashed,#3] (#1) to[bend right=30] (#2);
\end{pgfonlayer}
}
\newcommand{\EdgeL}[3]%
{ \begin{pgfonlayer}{background}
\draw[dashed,#3] (#1) to[bend left=30] (#2);
\end{pgfonlayer}
}
\newcommand{\EdgeLL}[3]%
{ \begin{pgfonlayer}{background}
\draw[dashed,#3] (#1) to[bend left=90] (#2);
\end{pgfonlayer}
}
\newcommand{\Arrow}[3]%
{ \begin{pgfonlayer}{background}
\draw[->,#3] (#1) -- (#2);
\end{pgfonlayer}}

\newcommand{\ArrowR}[3]%
{ \begin{pgfonlayer}{background}
\draw[->,#3] (#1) to[bend right=30] (#2);
\end{pgfonlayer}
}
\newcommand{\ArrowL}[3]%
{ \begin{pgfonlayer}{background}
\draw[->,#3] (#1) to[bend left=45] (#2);
\end{pgfonlayer}
}

\AtEndDocument{\refstepcounter{theorem}\label{finalthm}}

\title{Causal mediation analysis for stochastic interventions}

\date{\today}

\author[1]{Iv\'an D\'iaz \thanks{corresponding author:
    ild2005@med.cornell.edu}}
\author[2,3]{Nima S.~Hejazi}
\affil[1]{\small Division of Biostatistics,
  Weill Cornell Medicine.}
\affil[2]{\small Graduate Group in Biostatistics,
  University of California, Berkeley.}
\affil[3]{\small Center for Computational Biology,
  University of California, Berkeley.}

\begin{document}
\maketitle

\begin{abstract}
  Mediation analysis in causal inference has traditionally focused on
  binary exposures and deterministic interventions, and a
  decomposition of the average treatment effect in terms of direct and
  indirect effects. In this paper we present an analogous
  decomposition of the \textit{population intervention effect},
  defined through stochastic interventions on the exposure. Population
  intervention effects provide a generalized framework in which a
  variety of interesting causal contrasts can be defined, including
  effects for continuous and categorical exposures. We show that
  identification of direct and indirect effects for the population
  intervention effect requires weaker assumptions than its average
  treatment effect counterpart, under the assumption of no
  mediator-outcome confounders affected by exposure. In particular,
  identification of direct effects is guaranteed in experiments that
  randomize the exposure and the mediator. We discuss various
  estimators of the direct and indirect effects, including
  substitution, re-weighted, and efficient estimators based on
  flexible regression techniques, allowing for multivariate
  mediators. Our efficient estimator is asymptotically linear under a
  condition requiring $n^{1/4}$-consistency of certain regression
  functions. We perform a simulation study in which we assess the
  finite-sample properties of our proposed estimators. We present the
  results of an illustrative study where we assess the effect of
  participation in a sports team on BMI among children, using
  mediators such as exercise habits, daily consumption of snacks, and
  overweight status.
\end{abstract}

\section{Introduction}
Mediation analysis is a powerful analytical tool that allows
scientists to unveil the mechanisms through which causal effects
operate. The development of tools for mediation analysis has a long
history in the statistical sciences, starting with the early work of
\cite{wright1921correlation, wright1934method} on path analysis, which
provided the foundations for the later development of mediation
analysis using structural equation models
\citep{goldberger1972structural}. Indeed, one of the most widely used
mediation analysis methods is based on structural equations
\citep{baron1986moderator}. Recent decades have seen a revolution in
the field of causal inference from observational and randomized
studies, starting with the seminal work of \cite{Rubin74} on the
potential outcomes framework, which is itself rooted in ideas dating
back to \cite{Neyman23}. More recently, \cite{pearl1995causal,Pearl00}
has developed a causal inference framework using non-parametric
structural equation models, directed acyclic graphs, and the so-called
do-calculus. Related approaches have been proposed by \cite{Robins86},
\cite{spirtes2000causation}, \cite{dawid2000causal}, and
\cite{richardson2013single}. These frameworks allow researchers to
define causal effects non-parametrically, and to assess the conditions
under which causal effects can be identified from data. In particular,
novel tools have uncovered important limitations of the earlier work
on parametric structural equation models for mediation analysis
\citep{pearl1998graphs,imai2010general}. Essentially, structural
equation models impose implausible assumptions on the data generating
mechanism, and are thus of limited applicability to complex phenomena
in biology, health, economics, and the social sciences. For example,
modern causal models have revealed the incorrectness of the widely
popular method of \cite{baron1986moderator} in several important
cases, such as in the presence of confounders of the mediator-outcome
relationship \citep{cole2002fallibility}.

Using the potential outcomes framework, \cite{Robins&Greenland92}
introduced a non-parametric decomposition of the causal effect of a
binary exposure into so-called natural indirect and direct
effects. The indirect effect quantifies the effect on the outcome
through the mediator and the direct effect quantifies the effect
through all other mechanisms. \cite{Pearl01} arrived at an equivalent
effect decomposition using non-parametric structural equation
models. The identification of these natural (in)direct effects relies
on so-called cross-world counterfactual independencies, i.e.,
independencies on counterfactual variables indexed by distinct
hypothetical interventions. An important consequence of this
definition is that the natural (in)direct effect is not identifiable
in a randomized trial, which is problematic as it implies that
scientific claims obtained from these models are not falsifiable
through experimentation
\citep{Popper34,dawid2000causal,robins2010alternative}.

In an attempt to solve these problems, several authors have proposed
methods that do away with cross-world counterfactual
independencies. These methods can be divided in two types:
identification of bounds
\citep{robins2010alternative,tchetgen2014bounds,miles2015partial}, and
alternative definitions of the (in)direct effect
\citep{petersen2006estimation,van2008direct,vansteelandt2012natural,
  vanderweele2014effect}. Here, we take the second approach, defining
the (in)direct effect in terms of a decomposition of the total effect
of a stochastic intervention on the population exposure. 

Most causal inference problems consider deterministic interventions
that set each unit's exposure to some fixed value that could be a
function of the unit's baseline variables. Stochastic interventions
are a generalization of this framework, and are loosely defined as
interventions which yield an exposure that is a random variable after
conditioning on baseline variables. Estimation of total effects of
stochastic interventions was first considered by
\cite{stock1989nonparametric} and has been the subject of recent study
\citep{robins2004effects, didelez2006direct, tian2008identifying,
  Pearl2009, taubman2009intervening, stitelman2010impact,
  diaz2013assessing, dudik2014doubly, Haneuse2013,
  young2014identification}. Particularly relevant to this work are the
methods of \cite{Diaz12,Haneuse2013} who define total effects for
modified treatment policies, and \cite{kennedy2018nonparametric}, who
study identification and estimation of the total the effect of
propensity score interventions that shift a binary exposure
distribution. These papers do not address decomposition of the effects
of stochastic interventions on the exposure into direct and indirect
effects, which is the central theme of our manuscript.

Our methods are also related to a family of new direct and indirect
effects
\citep{didelez2006direct,vanderweele2014effect,lok2016defining,
  vansteelandt2017interventional,zheng2017longitudinal,rudolph2017robust,lok2019causal},
which have been collectively termed \textit{interventional effects}
\citep{nguyen2019clarifying}. This family of effects deals with binary
exposures and deterministic interventions on the exposure, and is thus
not entirely related to our approach, which deals with both continuous
and categorical exposures and stochastic interventions on the
exposure. Like the effects on the treated of
\cite{vansteelandt2012natural}, interventional effects share the
no-cross-world-independence property of our methods. The interested
reader is referred to \cite{nguyen2019clarifying} for a taxonomy of
the several mediation analyses proposed in the causal inference
literature up to date.

Stochastic interventions have analytical advantages compared to their
deterministic counterparts, such as allowing the seamless definition
of causal effects for continuous exposures with an interpretation that
is familiar to regular users of linear regression adjustment. For
example, \cite{Haneuse2013} assess the effect of an intervention that
reduces a patient's operating time (i.e., the time spent in surgery)
on the risk of post-operative outcomes among patients undergoing
surgical resection non-small-cell lung cancer.  \cite{Diaz12} study
the effect of increasing the amount of leisure time physical activity
in the elderly on subsequent all-cause
mortality. \cite{diaz2013assessing} study the effect of a
(hypothetical) policy that enforces pollution levels below a certain
cutoff point. \cite{kennedy2018nonparametric} shows that stochastic
interventions can also be used in longitudinal studies to define and
estimate total effects without relying on the positivity assumption.

In this article, we propose a decomposition of the effect of a
stochastic intervention into a direct and an indirect effect, with
interpretation analogous to that originally proposed by
\cite{Robins&Greenland92} and \cite{Pearl01}. We show that the
identification of (in)direct effects based on stochastic interventions
does not require cross-world counterfactual independencies, therefore
yielding scientific results that can be tested through experimentation
on both the exposure and mediator. 
Of high practical relevance, our
proposal also allows the definition and estimation of non-parametric
mediated effects for continuous exposures, a problem for which no
methods or software exist. Parametric mediation methods such as those
discussed by \cite{vansteelandt2012imputation} induce unquantifiable
amounts of bias by imposing untestable and implausible parametric
assumptions on the distribution of cross-world counterfactuals.

We develop a one-step non-parametric estimator based on the efficient
influence function, incorporating flexible regression tools from the
machine learning literature, and provide $n^{1/2}$-rate convergence
and asymptotic linearity results. We propose methods to use these
asymptotic distributions to construct confidence regions and to test
the null hypothesis of no direct effect. Our estimator has roots in
semiparametric estimation theory
\citep[e.g.,][]{pfanzagl1982contributions, begun1983information,
  vanderVaart91, newey1994asymptotic, Bickel97}, and in the targeted
learning framework of
\cite{vdl2006targeted,vanderLaanRose11,vanderLaanRose18}. In
particular, we use cross-fitting in order to obtain
$n^{1/2}$-convergence of our estimators while avoiding entropy
conditions that may be violated by the data adaptive estimators we use
\citep{zheng2011cross,chernozhukov2018double}. Our estimators use a
re-parameterization of certain integrals as conditional expectations
in order to accommodate multivariate mediators. Software implementing
our methods is provided in the form of an open source \texttt{R} package freely
available on GitHub.

\section{Mediation analysis for population intervention effects}
Let $A$ denote a continuous or categorical exposure variable, let $Y$
denote a continuous or binary outcome, let $Z$ denote a multivariate
mediator, and let $W$ denote a vector of observed covariates. Let
$O = (W, A, Z, Y)$ represent a random variable with distribution
$\P$. We use $\Pn$ to denote the empirical distribution of a sample of
$n$ i.i.d. observations $O_1, \ldots, O_n$. We let
$\P f = \int f(o)\dd \P(o)$ for a given function $f(o)$, and use $\E$
to denote expectations with respect to $\P$. We assume $\P \in \M$,
where $\M$ is the nonparametric statistical model defined as all
continuous densities on $O$ with respect to a dominating measure
$\nu$. Let $p$ denote the corresponding probability density
function. We use $g(a \mid w)$ to denote the probability density
function or the probability mass function of $A$ conditional on
$W = w$; $m(a,z,w)$ and $b(a,w)$ to denote the outcome regression
functions $\E(Y \mid A = a,Z = z,W = w)$ and
$\E(Y \mid A = a, W = w)$, respectively; and $e(a \mid z, w)$ to
denote the conditional density or probability mass function of $A$
conditional on $(Z, W)$. Let $g(a \mid w)$ be dominated by a measure
$\kappa(a)$ (e.g., the counting measure for binary $A$ and the
Lebesgue measure for continuous $A$). We use $q(z \mid a,w)$ and
$r(z \mid w)$ to denote the corresponding conditional densities of
$Z$. The parametrization $e = g q / r$ will prove fundamental in the
construction of our estimators, since it will allow us to avoid
estimation of multivariate conditional densities. A similar
parameterization is used by \cite{zheng2012targeted} to estimate
mediated effects under deterministic interventions. We use
$\cal W, \cal A, \cal Z$ and $\cal Y$ to denote the support of the
corresponding random variables.

We formalize the definition of our counterfactual variables using the
following non-parametric structural equation model (NPSEM), but note
that equivalent methods may be developed by taking the counterfactual
variables as primitives. Assume
\begin{equation}\label{eq:npsem}
  W = f_W(U_W);\ A = f_A(W, U_A);\
  Z = f_Z(W, A, U_M);\ Y = f_Y(W, A, Z U_Y).
\end{equation}
This set of equations represents a mechanistic model assumed to
generate the observed data $O$; furthermore, it encodes several
fundamental assumptions.  First, an implicit temporal ordering is
assumed --- that is, $Y$ occurs after $Z$, $A$ and $W$; $Z$ occurs
after $A$ and $W$; and $A$ occurs after $W$.  Second, each variable
(i.e., $\{W, A, Z, Y\}$) is assumed to be generated from the
corresponding deterministic function (i.e., $\{f_W, f_A, f_Z, f_Y\}$)
of the observed variables that precede it temporally, plus an
exogenous variable, denoted by $U$. Each exogenous variable is assumed
to contain all unobserved causes of the corresponding observed
variable. Independence assumptions on $U=(U_W,U_A,U_Z,U_Y)$ necessary
for identification will be clarified in
Section~\ref{sec:iden}. Furthermore, we note that we have explicitly
excluded outcome-mediator confounders which are affected by
exposure. Mediation analysis in the presence of a such variables is
notoriously hard \citep{avin2005identifiability}; the adaptation of
our methods to this problem is possible but it requires a new set of
tools which is out of the scope of this paper.

Causal effects are defined in terms of hypothetical interventions on
the NPSEM (\ref{eq:npsem}). In particular, consider an intervention in
which the equation corresponding to $A$ is removed, and the exposure
is drawn from a user-specified distribution $g_\delta(a \mid w)$,
which may depend on $g$ and is indexed by a user-specified parameter
$\delta$. We assume without loss of generality that
$g_{\delta=0}=g$. Let $A_\delta$ denote a draw from
$g_\delta(a \mid w)$. Alternatively, such modifications can sometimes
be described in terms of an intervention in which the equation
corresponding to $A$ is removed and the exposure is set equal to a
hypothetical regime $d(A, W)$. Regime $d$ depends on the natural (that
is, under no intervention) exposure level $A$ and covariates $W$. The
latter intervention is sometimes referred to as depending on the
\textit{natural value of exposure}, or as a \textit{modified
  treatment policy}
\citep{Haneuse2013}. \cite{young2014identification} provide a
discussion of the differences and similarities in the interpretation
and identification of these two interventions. Below, we discuss two
examples of stochastic interventions: modified treatment policies, and
exponential tilting.

\setcounter{example}{0}
\begin{example}[Modified treatment policy
  \citep{Haneuse2013}]\label{ex:1}
  Let $A$ denote a continuous exposure, such as operating time in
  non-small-cell lung cancer. Assume the distribution of $A$
  conditional on $W = w$ is supported in the interval $(l(w),
  u(w))$. That is, the minimum possible operating time for an
  individual with covariates $W = w$ is $l(w)$. Then one may define a
  hypothetical post-intervention exposure $A_\delta = d(A,W)$, where
  \begin{equation}\label{eq:defdshift}
    d(a, w) =
    \begin{cases}
      a - \delta & \text{if } a > l(w) + \delta \\
      a & \text{if } a \leq l(w) + \delta,
    \end{cases}
  \end{equation}
  where $0 < \delta < u(w)$ is an arbitrary user-given
  value. Interesting modifications to this regime may be obtained by
  allowing $\delta$ to be a function of $w$, therefore allowing the
  researcher to specify a different change in operating time as a
  function of covariates such as comorbidities, age, etc. This
  intervention was first introduced by \cite{Diaz12}, and has been
  further discussed in \cite{diaz2018stochastic} and
  \cite{Haneuse2013}.
\end{example}

\begin{example}[Exponential tilting]\label{ex:2}
  We can alternatively define a tilted intervention distribution as
\begin{equation}\label{eq:tilt}
  g_\delta(a \mid w) = \frac{\exp(\delta a) g(a \mid w)}
  {\int \exp(\delta a) g(a\mid w)\dd\kappa(a)},
\end{equation}
for $\delta \in \mathbb R$, and let the hypothetical post-intervention
exposure $A_\delta$ be a random draw from $g_\delta$, conditional on
the natural value of the observed covariates $W$.
For binary $A$, \cite{kennedy2018nonparametric} proposed evaluating
the total effect of a binary exposure $A$ in terms of incremental
propensity score interventions that replace the propensity score
$g(1 \mid w)$ with a shifted version based on multiplying the odds of
exposure by a user-given parameter $\delta'$. In particular, the
post-intervention propensity score is given by
\begin{equation}\label{eq:incremental}
  g_{\delta'}(1 \mid w) = \frac{\delta' g(1 \mid w)}{\delta' g(1 \mid w) + 1
  - g(1\mid w)},
\end{equation}
for $0 < \delta' < \infty$. The proposal of
\cite{kennedy2018nonparametric} is thus a case of exponential tilting
(\ref{eq:tilt}) under the parameterization $\delta' = \exp(\delta)$. This
choice of parameterization is motivated by the fact that $\delta'$ can
be interpreted as an odds ratio indicating how the intervention
changes the odds of exposure. The extremes of $\delta' = 0$
and $\delta' = \infty$ correspond to the standard interventions $A = 0$
and $A = 1$ considered in the definition of the average treatment
effect.
\end{example}

We now turn our attention to defining the \textit{population
  intervention effect (PIE)} of $A$ on $Y$. To proceed, for any values
$(a, z)$, consider the counterfactual outcome
$Y(a,z)=f_Y(W,a,z, U_Y).$, and the counterfactual mediator
$Z(a)=f_Z(W, a, U_Z)$. The counterfactual $Y(a,z)$ is the outcome in a
hypothetical world in which $(A,Z)=(a,z)$ is fixed externally.  The
PIE is defined as a contrast comparing the expectation of the outcome
under no intervention with the expectation of the counterfactual
outcome obtained under an intervention $A_\delta$:
\[\psi(\delta) = \E\{Y(A_\delta) - Y\}.\]
Note that the interpretation of the PIE depends on the stochastic
intervention considered. For example, for the modified treatment
policies of Example \ref{ex:1}, the PIE describes the difference in
outcomes obtained by a reduction of $\delta$ in operating time. In the
case of the incremental propensity score intervention
(\ref{eq:incremental}), the PIE is interpreted as the difference in
outcomes obtained by an intervention under which the odds of exposure
is $\delta'$ times higher compared to current practice.

Since $A$ is a cause of $Z$, an intervention that changes the exposure
to $A_\delta$ also induces a counterfactual mediator
$Z(A_\delta)$. 
As a consequence of the consistency implied by the NPSEM, we have
$Y(A,Z)=Y$. Similarly, the law of composition \citep{Pearl00} allows
us to write $Y(A_\delta,Z(A_\delta)) = Y(A_\delta)$.  Thus, the PIE
may be decomposed in terms of a \textit{population intervention direct
  effect (PIDE)} and a \textit{population intervention indirect effect
  (PIIE)}:
\begin{equation}\label{eqn:pie_decomp}
  \psi(\delta) = \overbrace{\E\{Y(A_\delta, Z(A_\delta)) -
    Y(A_\delta, Z)\}}^{\text{PIIE}} + \overbrace{\E\{Y(A_\delta, Z) -
    Y(A, Z)\}}^{\text{PIDE}}.
\end{equation}
This decomposition of the PIE as the sum of direct and indirect
effects has an interpretation analogous to the corresponding standard
decomposition of the average treatment effect \citep{Pearl01}.  In
particular, the direct effect represents the effect of an intervention
that changes the distribution of the exposure while keeping the
distribution of the mediators fixed at the value that it would have
taken under no intervention. The indirect effect measures the effect
of an indirect intervention on the mediators generated by intervening
on the exposure, while holding the intervention on the exposure
constant.

The intervention in Example \ref{ex:1} arises naturally as a modified
treatment policy. In contrast, the intervention in Example \ref{ex:2}
arises directly as a stochastic intervention that modifies the
distribution of the variables --- it is unclear as of yet whether this
quantity may be interpreted as a modified treatment policy.
Drawing on the work of \cite{Haneuse2013}, we make the following assumption for
modified treatment policies, which ensures that we can use the change of
variable formula when computing integrals over $\cal A$. This is useful for
studying properties of the parameter and estimators we propose.

\begin{assumptioniden}[Piecewise smooth invertibility]\label{ass:inv}
  For each $w \in \cal W$, assume that the interval
  ${\cal I}(w) = (l(w,), u(w))$ may be partitioned into subintervals
  ${\cal I}_{\delta,j}(w):j = 1, \ldots, J(w)$ such that $d(a, w)$ is equal to
  some $d_j(a, w)$ in ${\cal I}_{\delta,j}(w)$ and $d_j(\cdot,w)$ has inverse
  function $h_j(\cdot, w)$ with derivative $h_j'(\cdot, w)$.
\end{assumptioniden}
Under this assumption, the distribution of a modified treatment policy
$A_\delta=d(A,W)$ may be recovered through \citep[see][]{Haneuse2013}:
\begin{equation}\label{eq:gdelta}
  g_\delta(a \mid w) =
    \sum_{j=1}^{J(w)} I_{\delta, j} \{h_j(a, w), w\} g\{h_j(a, w)\mid w\}
    h_j'(a,w),
\end{equation}
where $I_{\delta, j} \{u, w\} = 1$ if $u \in {\cal I}_{\delta, j}(w)$
and $I_{\delta, j}\{u, w\} = 0$ otherwise. In Example~\ref{ex:1}, the
stochastic intervention becomes 
$$g_\delta(a\mid w) = g(a\mid w) \one\{l(w)\leq a \leq l(w) + \delta\} +
g(a+\delta\mid w) \one\{l(w)\leq a \leq u(w) - \delta\}.$$ Therefore,
under \ref{ass:inv}, a modified treatment policy may also be
represented as a change by which the equation $f_A$ is removed from
the NPSEM and $A$ is replaced by a draw $A_\delta$ from the
distribution $g_\delta(a\mid w)$. As a result of these two
representations, the intervention may be interpreted in two different
ways: (i) a change in the probabilistic mechanism used to assign
exposure level, and (ii) a subject-specific change in exposure from
$A$ to $A_\delta=d(A,W)$, where only interpretation (i) requires
\ref{ass:inv}. Note, however, that the population distribution of the
exposure is the same under both interventions
\citep{young2014identification}; thus, both representations lead to
exactly the same marginal counterfactual outcome distributions.

Several estimators of the functional $\psi(\delta)$ have previously
been proposed. For the case of a continuous exposure, \cite{Diaz12}
developed inverse probability weighted, outcome regression, and doubly
robust estimators based on the framework of targeted minimum
loss-based estimation (TMLE) \citep{vanderLaanRose11}, using data
adaptive estimators of the relevant nuisance
parameters. \cite{diaz2018stochastic} improved on the previous
methodology by constructing a TMLE algorithm with lower computational
complexity that preserves the desirable asymptotic properties of the
original approach.  \cite{Haneuse2013} propose estimators that rely on
correctly specified parametric models. Such methods are of limited
applicability since they are reliable only in situations where the
nuisance parameters involve only few categorical variables, where
correctly specified (that is, saturated) parametric models can
conscientiously be constructed. For the binary case with $g_\delta$ as
in Example \ref{ex:2}, \cite{kennedy2018nonparametric} proposed an
estimator for $\psi(\delta)$. This estimator is efficient,
asymptotically linear, and it allows incorporation of data adaptive
estimators of the nuisance parameters.

Since $\E(Y)$ is trivially estimated by the empirical mean in the
sample, our optimality theory and estimators focus on
$\theta(\delta) = \E\{Y(A_\delta, Z)\}$.  We present two types of
results: for general modified treatment policies satisfying
(\ref{ass:inv}), and for the particular stochastic intervention of
Example \ref{ex:2}. We compare the assumptions required for both.

\subsection{Identification}\label{sec:iden}
In this section we introduce the counterfactual variable $Y(a, z)$,
defined as the outcome that would be observed in a hypothetical world
in which $\P\{(A, Z)=(a, z)\} = 1$. This is the same counterfactual
variable that is often used to perform mediation analyses on the
average treatment effect \citep{Robins&Greenland92, Pearl01}.

We introduce the following identification assumptions:
\begin{assumptioniden}[Common support]\label{cond:cs}
Assume $\supp\{g_\delta(\, \cdot \mid w)\} \subseteq
\supp\{g(\, \cdot \mid w)\}$ for all $w\in \cal W$.
\end{assumptioniden}

\begin{assumptioniden}[Conditional exchangeability of exposure and
  mediator assignment]
\label{cond:cet}
Assume
\[\E\{Y(a, z) \mid A, W, Z\}=\E\{Y(a, z) \mid W, Z\}\text{ for all
  } (a,z)\in \cal A \times \cal Z.\]
\end{assumptioniden}

Assumption \ref{cond:cs} is standard in the analysis of causal
effects, and simply states that the $\delta$-specific intervention of
interest is supported in the data. This assumption holds for all
$\delta$ in the interventions described in Examples \ref{ex:1} and
\ref{ex:2} \citep[][]{Diaz12,kennedy2018nonparametric}. Assumption
\ref{cond:cet} is related to the assumption that
\cite{vansteelandt2012natural} used for identification of mediated
effects among the treated. In that proposal the authors assume
$Y(a, z) \indep (A, Z) \mid W$, which would imply the stronger
assumption $\E\{Y(a, z) \mid A, W, Z\}=\E\{Y(a, z) \mid W\}$. This
assumption would be satisfied for any pre-exposure variable $W$ in a
randomized experiment in which exposure and mediator are
randomized. Thus, the direct effect for a population intervention
corresponds to contrasts between treatment regimes of a randomized
experiment via interventions on $A$ and $Z$, unlike the natural direct
effect for the average treatment effect
\citep{robins2010alternative}. This claim is made rigorous in the
identification result of Theorem \ref{theo:iden} presented below. A
proof is available in the Supplementary Materials, together with the
assumptions on the NPSEM exogenous errors $U$ which are compatible
with \ref{cond:cet}.

\begin{theorem}[Identification]\label{theo:iden}
  Under \ref{cond:cs} and \ref{cond:cet}, $\theta(\delta)$ is
  identified and is given by
  \begin{equation}
    \theta(\delta) = \int m(a, z,
    w)g_\delta(a\mid w)p(z,w)\dd\nu(a,z,w).\label{eq:direct}
  \end{equation}
\end{theorem}

\begin{remark}[Mediator-outcome confounder not affected by exposure]\label{remark:na}
  Note that, like the natural direct effect of \cite{Pearl01}, we
  require that all confounders of the mediator-outcome relation are
  measured. This assumption is implicit in \ref{cond:cet}. To see why,
  consider the DAG in Figure~\ref{fig:dag}. Conditioning on the
  collider $Z$ opens a pathway from $A$ to $Y(a,z)$ through the
  outcome-mediator confounder $V$. If $V$ is not measured and adjusted
  for (i.e., $V\subseteq W$), then \ref{cond:cet} fails.

  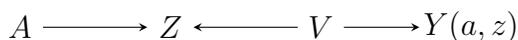
\begin{figure}[!htb]
    \centering
    \begin{tikzpicture}
      \node[minimum width=0.6cm,inner sep=0.05cm] (Z) at (0,0){$Z$};
      \node[minimum width=0.6cm,inner sep=0.05cm] (V) at (2,0){$V$};
      \node[minimum width=0.6cm,inner sep=0.05cm] (A) at (-2,0){$A$};
      \node[minimum width=0.6cm,inner sep=0.05cm] (Yaz) at (4,0){$Y(a,z)$};
      \draw[->,>=stealth,black] (A) -- (Z);
      \draw[->,>=stealth,black] (V) -- (Z);
      \draw[->,>=stealth,black] (V) -- (Yaz);
    \end{tikzpicture}
    \caption{Directed acyclic sub-graph of the variables involved in
      the case of an unmeasured mediator-outcome confounder.}
    \label{fig:dag}
  \end{figure}
\end{remark}

\begin{remark}[Mediator-outcome confounded by exposure]\label{remark:inter}
  The methods presented here cannot be used if the mediator-outcome
  confounder $V$ is affected by exposure. This is due to the
  introduction of a new counterfactual variable $V(a)$. In particular,
  consider the DAG in Figure~\ref{fig:dag2}, where we have included
  only the relevant factual and counterfactual variables. In this case,
  conditioning on the collider $V$ would open a path
  $A\rightarrow V \leftarrow U_V \rightarrow V(a) \rightarrow Y(a,z)$,
  and would make \ref{cond:cet} invalid. However, conditioning on $V$
  is necessary for \ref{cond:cet} in order to close the path
  $A\rightarrow Z \leftarrow V\rightarrow U_V \rightarrow V(a)
  \rightarrow Y(a,z)$, which gets open when we condition on the
  collider $Z$. A comprehensive discussion of issues in identification
  of path effects that includes this issue as a particular problem may
  be found in
  \cite{avin2005identifiability}. \cite{vanderweele2014effect} propose
  a solution to this problem which involves a stochastic intervention
  on the mediator $Z$. We note that this is intrinsically different
  from the problem treated here, since we are interested in stochastic
  interventions on $A$ (not on $Z$) and do not address
  mediator-outcome confounders affected by exposure.

  \begin{figure}[!htb]
    \centering
    \begin{tikzpicture}
      \node[minimum width=0.6cm,inner sep=0.05cm] (UV) at (0,-1){$U_V$};
      \node[minimum width=0.6cm,inner sep=0.05cm] (V) at (-1.5,-1){$V$};
      \node[minimum width=0.6cm,inner sep=0.05cm] (Va) at (1.5,-1){$V(a)$};
      \node[minimum width=0.6cm,inner sep=0.05cm] (A) at (-2.5,-2){$A$};
      \node[minimum width=0.6cm,inner sep=0.05cm] (Z) at (-0.5,-2){$Z$};
      \node[minimum width=0.6cm,inner sep=0.05cm] (Yaz) at (2.5,-2){$Y(a,z)$};
      \draw[->,>=stealth,black] (A) -- (V);
      \draw[->,>=stealth,black] (UV) -- (V);
      \draw[->,>=stealth,black] (V) -- (Z);
      \draw[->,>=stealth,black] (UV) -- (Va);
      \draw[->,>=stealth,black] (Va) -- (Yaz);
      \draw[->,>=stealth,black] (A) -- (Z);
    \end{tikzpicture}
    \caption{Directed acyclic sub-graph of the variables involved in the
      case of an outcome-mediator confounder affected by exposure.}
    \label{fig:dag2}
  \end{figure}
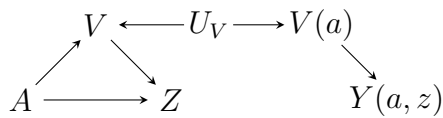
\end{remark}

\section{Optimality theory for estimation of the direct effect}\label{sec:optimal}
Thus far we have discussed the decomposition of the effect of a
stochastic intervention into direct and indirect effects, and
have provided identification results under weaker assumptions in
comparison to the natural direct effect. In the sequel, we turn our
attention to a discussion of efficiency theory for the estimation of
$\theta(\delta)$ in the nonparametric model $\M$. The
\textit{efficient influence function} (EIF) is a key object in
semi-parametric estimation theory, as it characterizes the asymptotic
behavior of all regular and efficient estimators \citep{Bickel97,
  van2002part}. Knowledge of the EIF has important practical
implications. First, the EIF is often useful in constructing locally
efficient estimators. There are three common approaches for this: (i)
using the EIF as an estimating equation
\citep[e.g.,][]{vanderLaan2003}, (ii) using the EIF in a one-step bias
correction \citep[e.g.,][]{pfanzagl1982contributions}, and targeted
minimum loss-based estimation \citep{vdl2006targeted, vanderLaanRose11,
vanderLaanRose18}. Second, the EIF estimating equation often enjoys desirable
properties such as multiple robustness, which allows for some components of the
data distribution to be inconsistently estimated while preserving
consistency of the estimator. Third, the asymptotic analysis of
estimators constructed using the EIF often yields second-order bias
terms, which require slow convergence rates (e.g., $n^{-1/4}$) for the
nuisance parameters involved, thereby enabling the use of flexible
regression techniques in estimating these quantities.

In Theorem~\ref{theo:eif} we present the EIF for a general stochastic
intervention. Although the components of the EIF associated with $Y$
and $(Z,W)$ are the same, the component associated with the model for
the distribution of $A$ must be computed on a case-by-case basis, that
is, for each intervention of interest. Proofs for all results are available in the
Supplementary Materials.

\begin{theorem}[Efficient influence function]\label{theo:eif}
  Let $\eta = (g, m, e)$. The efficient influence function for
  $\theta(\delta)$ in the nonparametric model $\mathcal M$ is
  $D^Y_{\eta, \delta}(o) + D^A_{\eta, \delta}(o) +
  D^{Z, W}_{\eta, \delta}(o) - \theta(\delta)$, where
  \begin{align*}
    D^Y_{\eta, \delta}(o) &= \frac{g_{\delta}(a \mid w)}{e(a \mid z, w)}
                            \{y - m(a,z, w) \} \\
    D^{Z, W}_{\eta, \delta}(o) &= \int m(a,z, w) g_{\delta}(a \mid w)
                                 \dd\kappa(a),
  \end{align*}
and $D^A_{\eta,\delta}(o)$ is the efficient score corresponding to the
non-parametric model for $g$.
\end{theorem}
An immediate consequence of Theorem \ref{theo:eif} is that, in a
randomized trial, we have $D^A_{\eta,\delta}(o)=0$.
Lemmas~\ref{lemma:mtp} and~\ref{lemma:tilt} below present the
$D^A_{\eta,\delta}(o)$ components for modified treatment policies
satisfying \ref{ass:inv} and for the exponential tilting of Example
\ref{ex:2}, respectively.

\begin{lemma}[Modified treatment policies]\label{lemma:mtp}
  Define the
  nuisance parameter
  \begin{align}
    \phi(a,w) &= \int m(d(a, w), z, w) r(z\mid w) d\nu(z)     \label{eq:paramr}\\
              &=\E \left\{\frac{g(A \mid W)}{e(A \mid Z, W)}
                m(d(A, W),Z, W) \mid A = a, W = w \right\},   \label{eq:paramnor} 
  \end{align} and augment $\eta$ as $\eta = (g, m, e, \phi)$. If the modified
  treatment policy $d(A,W)$ satisfies \ref{ass:inv}, then
  \[D^A_{\eta, \delta}(o) = \phi(a, w) - \int \phi(a, w) g(a \mid w)
    \dd\kappa(a).\]
\end{lemma}

\begin{lemma}[Exponential tilt]\label{lemma:tilt}
  Define the nuisance parameter
  \begin{align}
    \phi(a,w) &= \int m(a, z,w) r(z\mid w) d\nu(z)      \label{eq:paramr1}    \\
              &=\E \left\{\frac{g(A \mid W)}{e(A \mid Z, W)}
                m(A, Z,W) \mid A = a, W = w \right\},     \label{eq:paramnor1}
  \end{align}
  and augment $\eta$ as $\eta = (g, m, e, \phi)$. If the stochastic
  intervention is the exponential tilt (\ref{eq:tilt}), then
  $$D^A_{\eta, \delta}(o) = \frac{g_\delta(a \mid
    w)}{g(a \mid w)}\left\{\phi(a, w) - \int
    \phi(a, w)g_\delta(a\mid w)\dd\kappa(a)\right\}.$$
\end{lemma}


Expressions (\ref{eq:paramr}) and (\ref{eq:paramr1}) show that
estimators based on the respective influence functions require
integration with respect to the mediator $Z$, as well as estimation of
the possibly multivariate conditional density $r(z\mid w)$, which may
pose an estimation challenge due to the curse of dimensionality. To
solve the issue, we propose an alternative parametrization
(\ref{eq:paramnor}) and (\ref{eq:paramnor1}) of the EIF based on a
sequential regression $\phi$, rather than using the density of $Z$
conditional on $(A, W)$ and $W$. This choice has important
consequences for the purpose of estimation, as it helps to bypass
estimation of the (possibly high-dimensional) conditional density of
the mediators, instead allowing for regression methods, which are far
more commonly found in the statistics literature and software, to be
used for estimation of the relevant quantity. In particular, if
$r(z\mid w)$ is hard to estimate, estimators of $\phi$ may be computed
by first estimating $g$, $m$, and $e$, computing the pseudo-outcomes
defined in the lemmas, and applying regression techniques to estimate
the outer conditional expectation.

For binary exposures, the EIF corresponding to the incremental
propensity score intervention may be simplified as in the following
corollary.

\begin{coro}[Efficient influence function for incremental propensity
score interventions]\label{coro:tilt}
Let $A$ take values on $\{0, 1\}$, and let the exponentially tilted
intervention $g_{\delta,0}(1\mid W)$ be as in
(\ref{eq:incremental}). Then, the EIF of Lemma~\ref{lemma:tilt} may be
simplified as follows. Define the nuisance parameter
  $$\phi(w) = \E \left\{m(1, Z, W) - m(0, Z, W) \mid W = w \right\},$$
  and let $\eta = (g,m,e,\phi)$. Then
  $$D^A_{\eta,\delta}(o)= \frac{\delta\phi(w)\{a - g(1\mid w)\}}{\{\delta
      g(1\mid w) + 1 - g(1\mid w)\}^2}.$$
\end{coro}

Note that in Lemmas~\ref{lemma:mtp}, \ref{lemma:tilt}, and Corollary
\ref{coro:tilt}, we have used $\phi$ to represent different
parameters. We have allowed this abuse of notation because the nature
of this auxiliary parameter is the same for all three cases, and
having one symbol will allow us to state our estimation results in
some generality. In the sequel, the difference will always be clear
from context. Note also that $g(a\mid w)$ could be pulled out of the
expectation in the definition of $\phi(a,w)$. We decided to leave it
inside the expectation as we conjecture that it may act as a
stabilizing factor for the inverse probability weights
$\{e(a\mid z,w)\}^{-1}$.

In contrast to the efficient influence function for the natural direct
effect \citep{tchetgen2012semiparametric}, the contribution of the
exposure process to the EIF for the PIE mediated effect is
non-zero. This is a direct consequence of the fact that the parameter
of interest depends on $g$; moreover, this implies that, unlike the
natural direct effect, the efficiency bound in observational studies
differs from the efficiency bound in randomized studies. As we see in
the lemmas below, this also implies that it is not generally possible
to obtain estimating equations that are robust to inconsistent
estimation of $g$. However, such robustness will be possible if the
stochastic intervention is also a modified treatment policy satisfying
\ref{ass:inv}.

\begin{lemma}[Multiple robustness for modified treatment
  policies]\label{lemma:dr1}
  Let the modified treatment policy satisfy \ref{ass:inv}, and let
  $\eta_1=(g_1, e_1, m_1,\phi_1)$ be such that one of the two
  following conditions hold:
  \begin{enumerate}[label=(\roman*)]
  \item $g_1=g$ and either $e_1=e$ or $m_1=m$,
  \item $m_1=m$ and $\phi_1=\phi$.
  \end{enumerate}
  Then $\P D_{\eta_1,\delta}=\theta(\delta)$, with
  $D_{\eta,\delta}$ as defined in Theorem~\ref{theo:eif} and
  Lemma~\ref{lemma:mtp}.
\end{lemma}
The above lemma implies that it is possible to construct consistent
estimators for $\theta$ under consistent estimation of at least two of
the nuisance parameters in $\eta$, in the configurations described in
the lemma. This lemma is a direct consequence of Theorem
\ref{theo:dr1}, found in the Supplementary Materials. We note,
however, that part (ii) of the lemma may be uninteresting if the
parameterization (\ref{eq:paramnor}) is used to estimate $\phi$. In
that case $\phi_1=\phi$ will generally require $g_1=g$, $e_1=e$, and
$m_1=m$, as well as consistency of the estimator for the outer
expectation. In contrast, if the parameterization (\ref{eq:paramr}) is
used to estimate $\phi$, then the case $m_1=m$ and $\phi_1=\phi$ would
be trivially satisfied if $m_1=m$ and $r_1=r$, where $r_1$ is the
density used to compute $\phi_1$. To some readers it may seem
surprising that estimation of $\theta(\delta)$ may be robust to
estimation of $g$, even when the parameter definition in
(\ref{eq:direct}) is explicitly dependent on $g$. We offer some
intuition into this surprising result by noting that assumption
\ref{ass:inv} allows us to use the change of variable formula to
obtain
\[\theta(\delta) = \E\left\{\int m(d(A,W),z,W)r(z \mid, W)\dd\nu(z)\right\}.\]
Estimation of this parameter without relying on $g$ may be carried out
by consistently estimating $m$, $r$, and using the empirical
distribution as an estimator of the outer expectation. This behavior
has been previously observed for the total effect $\psi(\delta)$ under
\ref{ass:inv} \citep{Diaz12,Haneuse2013}.

The robustness properties of the EIF for an exponential tilt are
presented below.

\begin{lemma}[Robustness for exponential tilting]\label{lemma:dr2}
  Let $g_\delta$ be defined as in (\ref{eq:tilt}). Let
  $\eta_1 = (g_1, e_1, m_1, \phi_1)$ be such that $g_1 = g$ and either $e_1 = e$ or
  $m_1 = m$. Then $\P D_{\eta_1, \delta} = \theta(\delta)$, with
  $D_{\eta, \delta}$ as defined in Theorem~\ref{theo:eif} and
  Lemma~\ref{lemma:tilt}.
\end{lemma}

Lemma~\ref{lemma:dr2} is a direct consequence of Theorem \ref{theo:dr2}
in the Supplementary Materials. The corresponding proof reveals that the EIF for
the binary distribution is not multiply robust --- that is, the intervention
fails to satisfy assumption \ref{ass:inv} and integrals over the range of $A$
cannot be computed using change of variable formula. This behavior has been
previously observed for other interventions that do not satisfy \ref{ass:inv}
\citep{diaz2013assessing}. Even though this lemma implies that consistent
estimation of $g$ is required, the bias terms are still second-order, so an
estimator of $g$ converging at rate $n^{1/4}$ or faster is sufficient, as we
will see in the sequel.
\section{Estimation and statistical inference}\label{sec:estima}
We start this section describing two simple estimators, the substitution and
re-weighted estimators. These estimators are motivated by the fact that
$\theta(\delta)$ has the two following alternative representations:
\begin{eqnarray}
  \theta(\delta) &=& \E\left\{\int m(a, Z, W)
      g_\delta(a \mid W)\dd\kappa(a) \right\}\label{eq:gcomp}\\
    &=& \E \left\{\frac{g_\delta(A \mid W)}
      {e(A \mid Z, W)}\, Y \right\},\label{eq:ipw}
\end{eqnarray}
where we remind the reader that $e(a \mid z, w)$ denotes the
probability density function of $A$ conditional on $(Z, W)$. Equation
(\ref{eq:ipw}) follows from noting that $g q/r = e$.  This
parameterization has the advantage that only the univariate
conditional density $e(a\mid z, w)$ has to be estimated, instead of
the conditional densities of the possibly high-dimensional mediator
$Z$. A similar result was also used by \cite{zheng2012targeted} to
develop a targeted minimum loss-based estimator of natural direct
effects under a binary exposure variable.

The substitution estimator is simply defined by plugging in estimators
of $m$ and $g_\delta$ into the identification result given in
(\ref{eq:gcomp}).  Consistency of this estimator requires consistent
estimation of the outcome regression $m$ and the intervention
distribution $g_\delta$. The second estimator is a re-weighting
estimator based on the alternative representation of the
identification result given in (\ref{eq:ipw}), which requires
consistent estimation of $g_\delta$ and $e$. In the remainder of this
section, we discuss an efficient estimator that combines ideas from the
previous two estimators as well as the efficient influence function
derived in the previous section, in order to build an estimator that
is both efficient and robust to model misspecification. We discuss an
asymptotic linearity result for the doubly robust estimator that
allows computation of asymptotically correct confidence intervals and
hypothesis tests.

In the sequel, we assume that preliminary estimators $\hat{m}$,
$\hat{g_\delta}$, $\hat\phi$ and $\hat{e}$ of $m$, $g_\delta$, $\phi$,
and $e$, respectively, are available. These estimators may be obtained
from flexible regression techniques such as support vector machines,
regression trees, boosting, neural networks, splines, or ensembles thereof
\citep{Breiman1996, vanderLaan&Polley&Hubbard07}. As previously discussed, the
consistency of these estimators will determine the consistency of our
estimators of the population mediation intervention mean $\theta$.

\subsection{Substitution estimator and re-weighted estimators}
First, we discuss a substitution estimator based on (\ref{eq:gcomp}), computed
as
\begin{equation}\label{eq:est_sub}
\thetasub = \int \frac{1}{n} \sum_{i = 1}^n \hm(a, Z_i, W_i) \hgd(a \mid W_i)
  \dd\kappa(a),
\end{equation}
where we have substituted estimators of $m$ and $g_\delta$ in
(\ref{eq:gcomp}), and have estimated the expectation with respect to the joint
density $p(z, w)$ by the empirical mean. The re-weighted estimator is based on
(\ref{eq:ipw}), and is defined by
\begin{equation}\label{eq:est_re}
\thetare = \frac{1}{n} \sum_{i = 1}^n\frac{\hat g_{\delta}(A_i \mid W_i)}
{\hat{e}(A_i \mid, Z_i, W_i)} Y_i
\end{equation}

If $\hm$, $\hgd$, and $\hat{e}$ are estimated within parametric
models, then, by the delta method, both $\thetasub$ and $\thetare$ are
asymptotically linear and $n^{1/2}$-consistent. The bootstrap or an
influence function-based estimator may be used to construct
asymptotically correct confidence intervals. However, if either the
mediators or confounders are high-dimensional, the required
consistency of $\hm$, $\hgd$, and $\hat{e}$ will hardly be achievable
within parametric models. This issue may be alleviated through the use
of data adaptive estimators. Unfortunately, $n^{1/2}$-consistency of
$\thetasub$ and $\thetare$ will generally require that $\hm$, $\hgd$,
and $\hat{e}$ are consistent in $L_2(\P)$-norm at parametric rate,
which is generally not possible within data adaptive estimation of
high-dimensional regressions. Thus, the asymptotic distribution will
generally be unknown, rendering the construction of confidence
intervals and hypothesis tests impossible. In the following section,
we use the efficient influence function to propose an estimator that
is $n^{1/2}$-consistent and efficient under a weaker assumption,
requiring only $n^{1/2}$-convergence of second-order regression bias
terms.

\subsection{Efficient estimator}

We propose using the efficient influence function $D_{\eta, \delta}$
to construct a robust and efficient estimator, constructed as the
solution to the estimating equation $\Pn D_{\hat \eta, \delta} = 0$ in
$\theta$, for a preliminary estimator $\hat \eta$ of $\eta$. In order to avoid imposing entropy conditions on the initial
estimators, we advocate for the use of cross-fitting
\citep{zheng2011cross, chernozhukov2016double} in the estimation
procedure.  Let ${\cal V}_1, \ldots, {\cal V}_J$ denote a random
partition of the index set $\{1, \ldots, n\}$ into $J$ prediction sets
of approximately the same size. That is,
${\cal V}_j\subset \{1, \ldots, n\}$;
$\bigcup_{j=1}^J {\cal V}_j = \{1, \ldots, n\}$; and
${\cal V}_j\cap {\cal V}_{j'} = \emptyset$. In addition, for each $j$,
the associated training sample is given by
${\cal T}_j = \{1, \ldots, n\} \setminus {\cal V}_j$. Denote by
$\hat \eta_{j}$ the estimator of $\eta = (g, m, e, \phi)$, obtained by
training the corresponding prediction algorithm using only data in the
sample ${\cal T}_j$. Further, let $j(i)$ denote the index of the
validation set which contains observation $i$. The estimator is thus
defined as:
\begin{equation}\label{eq:aipw}
  \thetaaipw = \frac{1}{n} \sum_{i = 1}^n D_{\hat\eta_{j(i)}, \delta}(O_i) =
    \frac{1}{n}\sum_{i = 1}^n \left\{D^Y_{\hat\eta_{j(i)}, \delta}(O_i) +
    D^A_{\hat\eta_{j(i)}, \delta}(O_i) + D^{Z,W}_{\hat\eta_{j(i)},
    \delta}(O_i) \right\}.
\end{equation}

In a randomized trial the estimator may also be computed by setting
$ D^A_{\hat\eta_{j(i)}, \delta}(O_i)=0$. $M$-estimation theory may be
used to derive the asymptotic distribution of $\thetaaipw$. Asymptotic
linearity and efficiency of the estimator for modified treatment
policies is detailed in the following theorem:

\begin{theorem}[Pointwise weak convergence for modified treatment
  policies]\label{ass:lin1}
  Let $\norm{\cdot}$ denote the $L_2(\P)$-norm defined as
  $\norm{f}^2 = \int f^2 \dd \P$. Assume
  \begin{enumerate}[label=(\roman*)]
  \item
    $\norm{\hat{m} - m} \left\{\norm{\hat{g} - g}+ \norm{\hat{e}
      - e} \right\} + \norm{\hat{g} - g}\, \norm{\hat{\phi} - \phi} =
      o_\P(n^{-1/2})$\label{ass:sec1}, and
  \item $\P\{|D_{\eta, \delta}(O)| \leq C \} = \P \{|
    D_{\hat{\eta}, \delta}(O) | \leq C \} = 1$ for some $C < \infty$, and
    \label{ass:bounded}
  \item The modified treatment policy $d(a,w)$ is piecewise smooth
    invertible (\ref{ass:inv}).
  \end{enumerate}
Then
  $$\sqrt{n}\{\thetaaipw - \theta(\delta)\} \rightsquigarrow N(0,
    \sigma^2(\delta)),$$ where
  $\sigma^2(\delta) = \var\{D_{\eta, \delta}(O) \}$ is the efficiency bound.
\end{theorem}

Theorem~\ref{ass:lin1} establishes the weak convergence of
$\thetaaipw$ pointwise in $\delta$. This convergence is useful to
derive confidence intervals in situations where the modified treatment
policy has a suitable scientific interpretation for a given $\delta$,
such as in our Example~\ref{ex:1}. Under the assumptions of the
theorem, an estimator $\hat\sigma^2(\delta)$ of $\sigma^2(\delta)$ may
be obtained as the empirical variance of
$D_{\hat\eta_{j(i)},\delta}(O_i)$, and a Wald-type confidence interval
may be constructed as
$\thetaaipw\pm z_{1-\alpha/2} \hat\sigma(\delta)/\sqrt{n}$.

For the remainder of this section, we turn our attention to a
discussion of uniform convergence of $\thetaaipw$. Such a convergence
result will prove useful in the following section, where we establish
a hypothesis test of no direct effect. Such a test is constructed by
rejecting the hypothesis if the direct effect is non-significant (at
level $\alpha$), uniformly in $\delta$. To build such a testing
procedure, we focus on the intervention defined in terms of
exponential tilting (\ref{eq:tilt}). Results for modified treatment
policies are possible as well; however, these require smoothness
assumptions on the map $\delta \mapsto g_\delta(a \mid w)$. Inspection
of (\ref{eq:gdelta}) reveals that this may in turn require smoothness
assumptions on $a \mapsto g(a \mid w)$, which may not be justifiable
in a number of applications. We thus focus on exponential tilting,
which yields smooth maps $\delta \mapsto g_\delta(a \mid w)$ by
construction. This discussion, together with Lemmas~\ref{lemma:mtp}
and \ref{lemma:tilt}, thus reveals a trade-off between smoothness and
robustness in estimation of modified treatment policies and
exponential tilting.

\begin{theorem}[Uniform weak convergence for exponential
  tilting]\label{ass:lin2}
  Let $g_\delta$ be the exponential tilting intervention distribution
  (\ref{eq:tilt}) and let $\Delta = [\delta_l, \delta_u]$ denote an
  interval with $0 < \delta_l \leq \delta_u < \infty$.  Define
  $c(w) = \{\int_a \exp(\delta a)g(a\mid w)\}^{-1}$. Assume
  $||\hat{c} - c||^2 = o_\P(n^{-1/2})$ as well as \ref{ass:sec1} and
  \ref{ass:bounded} stated in Theorem~\ref{ass:lin1}. Then
  \[\sqrt{n}\{\thetaaipw - \theta(\delta)\} \rightsquigarrow \mathbb{G}(\delta)\]
  in $\ell^\infty(\Delta)$, where for any
  $\delta_1,\delta_2\in \Delta$, $\mathbb{G}(\cdot)$ is a mean-zero
  Gaussian process with covariance function
  $\E \{\mathbb{G}(\delta_1) \mathbb{G}(\delta_2) \} = \E\{D_{\eta,
    \delta_1} (O)D_{\eta, \delta_2}(O)\}$.
\end{theorem}

\subsection{Uniform inference and tests for the hypothesis of no
  direct effect}\label{sec:node}
In this section, we consider estimation of the direct effect
$\beta(\delta) = \theta(\delta)-\E(Y)$. Define the corresponding
(uncentered) influence function $S_{\eta, \delta}(o) = D_{\eta,
\delta}(o) - y$. A straightforward extension of Theorem~\ref{ass:lin2} shows that
$\hat{\beta}(\delta) = \hat\theta(\delta)-\bar{Y}$ converges weakly to a
process ${\mathbb{G}}(\delta)$ with covariance function
$\E\{\mathbb{G}(\delta_1) \mathbb{G}(\delta_2)\} = \E\{S_{\eta, \delta_1}(O)
S_{\eta, \delta_2}(O) \}$.

We now present an approach to constructing uniform confidence bands on
the function $\beta(\delta)$, allowing testing of the null hypothesis
of no direct effect $H : \sup_{\delta \in \Delta} \beta(\delta) =
0$. This hypothesis test is useful for checking the existence of a
direct effect even if the interpretation of the exponential tilt
$g_\delta$ (e.g., as the odds ratio comparing post vs pre-intervention
odds of exposure) does not answer a particularly meaningful question
in a given application. Let $\hat\sigma(\delta)$ denote the empirical
variance of $S_{\hat \eta_{j(i)}, \delta}(O_i)$. Our goal will be
achieved by finding a value $c_\alpha$ such that
$\hat \rho(c_\alpha) = 1 - \alpha$, where $\hat\rho$ is a function
such that
\begin{equation}\label{eq:apmb}
  \hat\rho(t) = \prob\left(\sup_{\delta \in \Delta} \bigg| \frac{\hat\beta(\delta)
      - \beta(\delta)}{\hat\sigma(\delta) / \sqrt{n}}\bigg| \leq t \right) +
      o_\P(1).
\end{equation}
Confidence bands may be computed as
$\hat\theta \pm n^{-1/2}c_\alpha\hat\sigma(\delta)$, and p-values for
$H$ can be computed by evaluating $1 - \hat\rho(t)$ at the observed
value of the supremum test statistic. The function $\hat\rho(t)$ may
be obtained by approximating the distribution of
$\sup_{\delta \in \Delta} {\mathbb{G}(\delta)}$, where
${\mathbb{G}(\delta)}$ is the Gaussian process defined above. In this
paper we take the approach proposed by
\cite{kennedy2018nonparametric}, using the multiplier bootstrap
\citep{gine1984some,vanderVaart&Wellner96,chernozhukov2013gaussian,
  belloni2015uniformly}. We omit the relevant proofs as they are
identical to those presented by \cite{kennedy2018nonparametric}. In
comparison with the nonparametric bootstrap, the multiplier bootstrap
has the computational advantage that the nuisance estimators
$\hat\eta$ need not be re-estimated. In comparison with directly
sampling $\sup_{\delta \in \Delta} {\mathbb{G}(\delta)}$, the proposed
procedure does not require the evaluation of potentially large
covariance matrices; therefore, it is far more computationally
efficient and convenient.

The multiplier bootstrap approximates the distribution of
$\sup_{\delta \in \Delta} {\mathbb{G}(\delta)}$ with the supremum of the process
$${\mathbb{M}}(\delta) = \frac{1}{\sqrt{n}}\sum_{i = 1}^n \frac{\xi_i
  \{S_{\hat\eta_{j(i)}, \delta}(O_i) - \hat\beta(\delta)\}}
  {\hat\sigma(\delta)},$$
where randomness is introduced through sampling the multipliers
$(\xi_1, \ldots, \xi_n)$, despite the process being conditional on the observed
data $O_1, \ldots, O_n$. The multiplier variables are i.i.d.~with mean zero and
unit variance, and are drawn independently from the sample. Typical choices are
Rademacher ($\prob(\xi = -1) = \prob(\xi=1)=0.5$) or Gaussian multipliers. Under
the assumptions of Theorem~\ref{ass:lin2}, plus uniform consistency of
$\hat\sigma(\delta)$, it can be shown that (\ref{eq:apmb}) holds for
$$\hat\rho(t) = \prob \left(\sup_{\delta \in \Delta} \big|{\mathbb{M}}(\delta)
  \big|\leq t\,\bigg|\, O_1, \ldots, O_n\right).$$
As a consequence, computation of the critical value, p-values, and confidence
intervals only requires simulation of a large number of realizations of the
multipliers over a fine grid over $\Delta$.

\section{Simulation study}
We now turn to comparing the three estimators of the direct
effect, previously considered in Section \ref{sec:estima}. In
particular, we investigate the performance of the substitution
(\ref{eq:gcomp}, \ref{eq:est_sub}), re-weighted (\ref{eq:ipw},
\ref{eq:est_re}), and efficient (\ref{eq:aipw}) estimators in the case
of an incremental propensity score (IPS) intervention on a binary intervention variable of interest. The estimators are evaluated on
data simulated from the following data-generating mechanism:
\begin{align*}
  W_1 &\sim \bern(0.50);
  W_2 \sim \bern(0.65);
  W_3 \sim \bern(0.35) \\
  A &\sim \bern\left(\frac{1}{4} \cdot \sum_{j = 1}^3 W_j +
    0.1\right) \\
  Z_1 &\sim \bern\left(1 - \expit\left[\frac{A + W_1}{A +
        W_1 + 0.5}\right]\right) \\
  Z_2 &\sim \bern\left(\expit\left[\frac{(A - 1) +
        W_2}{W_3 + 3}\right]\right) \\
  Z_3 &\sim \bern\left(\expit\left[\frac{(A - 1) + 2
        \cdot W_1 - 1}{2 \cdot W_1 + 0.5}\right]\right)\\
  Y &= Z_1 + Z_2 - Z_3 + A - 0.1 \cdot \left(\sum_{j = 1}^3 W_j\right)^2 +
    \epsilon,
\end{align*}
where $\bern(p)$ is the Bernoulli distribution with parameter $p$,
$\expit(x) = \{1 + \exp(-x)\}^{-1}$ is the CDF of the logistic
distribution (as implemented in the \texttt{plogis} function in the
\texttt{R} programming language), and
$\epsilon \sim \text{N}(0,0.25)$.  The data available on a single
observational unit is denoted by the random variable
$O = (W_1, W_2, W_3, A, Z_1, Z_2, Z_3, Y)$, where, in any given
simulation, we consider observing $n$ i.i.d.~copies of $O$ for one of
seven sample sizes $n \in \{400, 900, 1600, 2500, 3600, 4900, 6400\}$.

Under the above data-generating mechanism, we seek to estimate the
direct effect under an incremental propensity score intervention
$\delta = 0.5$, for which the true value of the natural direct effect
is approximately $0.137$. We approximated this effect by using the
alternative representation of $\theta(\delta)$ as
\[\theta(\delta) = \E\left\{\int m(a,Z,W)g_\delta(a\mid
    W)\dd\nu(a)\right\},\] where the inner integral is approximated by
Monte Carlo integration through a large sample $a_1,\ldots,a_m$ of
uniformly distributed numbers in the range of $A$, and the outer
expectation is approximated through the law of large numbers by drawing a
large sample $(W_1,Z_1), \ldots, (W_k, Z_k)$ from the joint
distribution of $(W,Z)$. Each of the estimators is evaluated by
contrasting regimes in which the appropriate nuisance parameters are
fit via a well-specified nonparametric regression or misspecified by
fitting an intercept model. The enumerated set of estimators and
regimes is summarized in Table~\ref{table:est_perform} and
Figure~\ref{fig:all_metric_comparison}. In order to ensure a
well-specified nonparametric regression for the nuisance parameters,
we rely on the highly adaptive lasso (HAL) estimator, a recently
proposed nonparametric regression function with properties
guaranteeing convergence of estimated nuisance components at the
$n^{1/4}$-rates required by our theorems \citep{benkeser2016highly,
  vdl2017generally, vdl2018highly}.

\begin{table}[H]
  \centering
  \begin{tabular}{r|c|c|c|c|c|c|c}\hline
    & \multicolumn{7}{c}{$n$}\\\hline
    \textit{Estimator} & $400$ & $900$ & ${1600}$ & ${2500}$ &
    ${3600}$ & ${4900}$ & ${6400}$\\
    \hline
    Substitution       & 0.083 & 0.086 & 0.084 & 0.077 & 0.072 & 0.074 & 0.075  \\
    Reweighted (IPW)   & 0.105 & 0.120 & 0.111 & 0.116 & 0.107 & 0.112 & 0.109  \\
    Efficient          & 0.092 & 0.086 & 0.071 & 0.072 & 0.068 & 0.067 & 0.065  \\
    Efficient (E mis.) & 0.060 & 0.060 & 0.060 & 0.059 & 0.054 & 0.059 & 0.058  \\
    Efficient (M mis.) & 0.165 & 0.130 & 0.110 & 0.107 & 0.099 & 0.103 & 0.097  \\
    Efficient (G mis.) & 0.436 & 0.829 & 1.255 & 1.912 & 2.662 & 3.543 & 4.519  \\
    \hline
  \end{tabular}
  \caption{Mean-squared errors (MSE), scaled by $n$, of the three key estimators of the direct effect
    under an IPS intervention $\delta = 0.5$, across $1000$ Monte Carlo
    simulations for each of seven sample sizes. Substitution and
    reweighted estimators are computed using HAL for $g$, $m$, and
    $e$. ``E mis.'' denotes that $e$ was inconsistently estimated via an
    intercept-only logistic regression model, ``M mis.'' and ``G mis.'' denote
    analogous estimators.}
  \label{table:est_perform}
\end{table}

\begin{figure}[H]
  \hspace*{-1.5cm}
  \centering
  \includegraphics[scale=0.4]{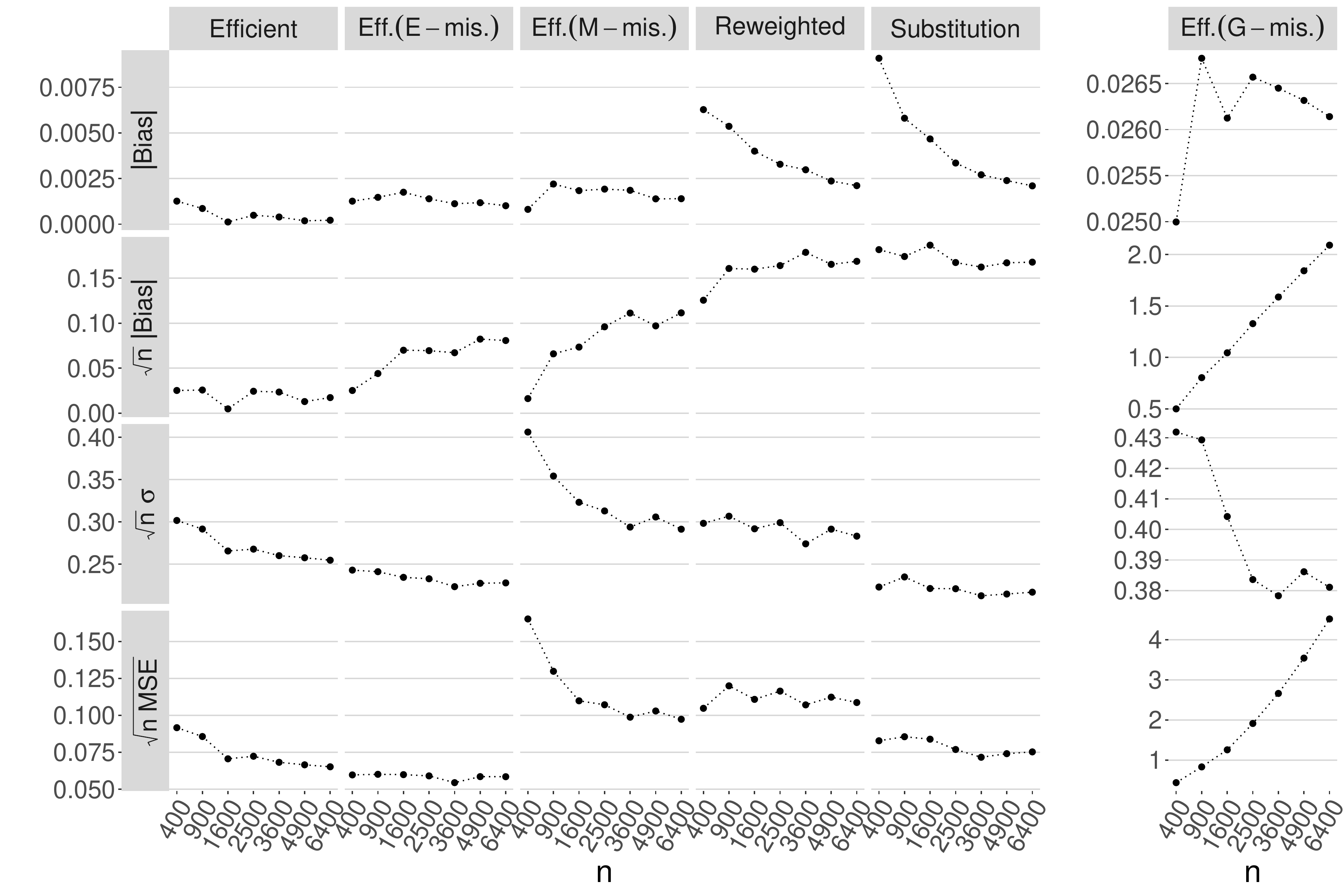}
  \caption{Statistics for the three key estimators (and variations
    thereof) of the direct effect under an IPS intervention
    $\delta = 0.5$, across $1000$ Monte Carlo simulations for each of
    seven sample sizes.}
  \label{fig:all_metric_comparison}
\end{figure}

Inspection of the mean-squared error, after scaling by $\sqrt{n}$,
reveals that the substitution estimator and the efficient one-step
estimator both display excellent, essentially equivalent performance
when nuisance components are estimated using the highly adaptive
lasso. The one-step estimator has slightly better performance, which
seems to be driven by a better bias-variance trade-off. In contrast to
the substitution estimator, the efficient one-step estimator carries
the advantage of being double robust, allowing misspecification of
either the outcome regression (denoted ``M'') or the
mediator-inclusive propensity score (denoted ``E''). The robustness of
the efficient estimator to the misspecification of these nuisance
components --- and the lack of robustness to the mediator-exclusive
propensity score (denoted ``G'') --- are demonstrated in the last
three rows of Table~\ref{table:est_perform}.
Figure~\ref{fig:all_metric_comparison} visualizes the performance of the
estimators, and their misspecified variants, in terms of both the MSE
(as presented in Table~\ref{table:est_perform}) and its individual
components, the bias and standard error. This comparison of the
estimators reveals that the correctly-specified one-step efficient
estimator displays excellent performance in terms of both bias and
variance while its non-robust misspecified variant displays an
asymptotic bias that grows with sample size. Interestingly, the
one-step estimator with $e$ inconsistently estimated displayed
better performance than the fully efficient version. This is possibly
an idiosyncrasy of this simulation due to the fact that
misspecification through an intercept-only model generates smaller
variable weights. Altogether, these numerical investigations
demonstrate the utility of the proposed estimators in settings where
the nonparametric estimation of nuisance components is viable;
moreover, in applied data analytic settings where this procedure may
be of interest, the one-step efficient estimator is clearly preferable on
account of its multiple robustness. All numerical studies of the
estimators were performed using the implementations available in the
\texttt{medshift} software package \citep{hejazi2019medshift} for the
\texttt{R} language and environment for statistical computing
\citep{R}.

\section{Application}

We now turn to considering a scenario in which the decomposition
proposed in equation \ref{eqn:pie_decomp} and the proposed efficient
estimator (\ref{eq:aipw}) may be used to estimate direct and indirect
effects. To proceed, we take as example a simple data set from an
observational study of the relationship between BMI and children's
behavior, distributed as part of the \texttt{mma R} package, available via the
Comprehensive \texttt{R} Archive Network
(\url{https://CRAN.R-project.org/package=mma}). The documentation of this data
set describes it as a ``database obtained from the Louisiana State University
Health Sciences Center, New Orleans, by Dr.~Richard Scribner [who] explored the
relationship between BMI and kids behavior through a survey at children,
teachers and parents in Grenada in 2014. This data set includes $691$
observations and $15$ variables.'' In particular, we consider a
modified version of this data set with all missing values removed, as
these are irrelevant to the demonstration of the proposed
methodology. In standard data analytic practice, we advocate for the use of the
proposed methodology in tandem with a correction for missing data,
such as imputation or weighting by inverse probability of censoring.
\citep{carpenter2006comparison,vansteelandt2010analysis,seaman2012combining}.

To demonstrate the assessment of the direct and indirect effect with
this observational data set, we consider the effect of participation
in a sports team on the BMI of children, taking several related
covariates as mediators (including snacking, exercising, and overweight status)
and all other collected covariates as potential confounders. As the
intervention variable is binary, we frame our proposal in terms of an
incremental propensity score intervention \citep{kennedy2018nonparametric},
wherein the odds of participating in a sports team is increased by a factor of
$\delta = 2$ for each individual. Such a stochastic exposure regime may be
interpreted as the introduction of a school program or policy that motivates
children to opt in to participating in a sports team, doubling the odds of such
voluntary participation.

\subsection{Estimation Strategy}

As noted in equation \ref{eqn:pie_decomp}, the population intervention effect
admits a decomposition in terms of components that allow estimation of the
direct and indirect effects. We compute each of the components of the direct
and indirect effects using appropriate estimators as follows

\begin{itemize}
  \itemsep0.25pt
\item for $\mathbb{E}\{Y(A, Z)\} = \mathbb{E}Y$, the natural value of
  the outcome under no intervention, the empirical mean in the sample serves
  as an efficient estimator;
\item for $\mathbb{E}\{Y(A_{\delta}, Z)\} = \theta(\delta)$, the mean
  outcome under an intervention altering the exposure mechanism but
  not the mediation mechanism, a one-step efficient estimator, denoted
  $\hat{\theta}(\delta)$, is proposed as equation \ref{eq:aipw} and made
  available via the \texttt{medshift R} package
  \citep{hejazi2019medshift};
\item for $\mathbb{E}\{Y(A_{\delta})\} = \psi(\delta)$,
  the mean outcome under an intervention altering both the exposure
  and mediation mechanisms, a one-step efficient estimator, denoted
  $\hat{\psi}(\delta)$ in the sequel, is easily estimable using the
  \texttt{npcausal R} package \citep{kennedy2018npcausal}.
\end{itemize}

In the construction of estimators for $\theta(\delta)$ and $\psi(\delta)$,
data adaptive nonparametric regression procedures are incorporated to allow the
relevant nuisance parameters of each estimator to be computed in a flexible
manner using various \texttt{R} packages. The \texttt{npcausal} package allows
the estimator $\hat{\psi}(\delta)$ to be constructed using the \texttt{ranger}
algorithm \citep{wright2015ranger}, an efficient and fast implementation of
random forests \citep{breiman2001random}. In constructing
$\hat{\theta}(\delta)$, the \texttt{medshift} package provides facilities for
estimating nuisance parameters data adaptively via the Super Learner algorithm
\citep{vanderLaan&Polley&Hubbard07} for constructing ensemble learners through
cross-validation, using its implementation in the  \texttt{sl3} package
\citep{coyle2018sl3}. In particular, the Super Learner procedure was used to
create a weighted ensemble of algorithms from a library including extreme
gradient boosting via the \texttt{xgboost} package \citep{chen2016xgboost}, with
variants including 50, 100, and 300 boosting iterations; variants of random
forests using 50, 100, and 500 trees; $L_1$-penalized lasso and $L_2$-penalized
ridge GLMs via the \texttt{glmnet} package \citep{friedman2009glmnet}; an
elastic net GLM with equally weighted $L_1$ and $L_2$ penalization terms (also
via \texttt{glmnet}); a main terms GLM; an intercept model; and the highly
adaptive lasso \citep{benkeser2016highly}, with 5--fold cross-validation and up
to either 3-way or 5-way interaction terms, using the \texttt{hal9001} package
\citep{coyle2018hal9001}.

\subsection{Estimating the Direct and Indirect Effects}

From the decomposition given in equation \ref{eqn:pie_decomp}, the
direct effect may be denoted
$\beta(\delta) = \theta(\delta)-\mathbb{E}Y$. An estimator of the
direct effect, $\hat{\beta}(\delta)$ may be expressed as a composition
of estimators of its constituent parameters:
\begin{equation*}
  \hat{\beta}({\delta}) = 
  \hat{\theta}(\delta)-\frac{1}{n} \sum_{i = 1}^n Y_i.
\end{equation*}
Using the estimation strategies previously outlined, we may construct
an estimate of the direct effect through a straightforward application
of the delta method, yielding both a point estimate and associated
standard errors under our proposed stochastic intervention
policy. Similarly, the indirect effect $\psi(\delta)-\theta(\delta)$
may be estimated as $\hat{\psi}(\delta)-\hat{\theta}(\delta)$. We
provide both point estimates and associated inference under our
proposed stochastic intervention policy in
Table~\ref{table:nde_nie_applic}.

\begin{table}[H]
  \centering
  \begin{tabular}{r|r|r|r}
    \hline
    \textit{Parameter} & Lower 95\% CI & Estimate & Upper 95\% CI\\
    \hline
    Direct Effect & -0.458 & 0.011 & 0.479\\
    \hline
    Indirect Effect & -0.672 & -0.157 & 0.357\\
    \hline
  \end{tabular}
  \caption{Point estimates and 95\% confidence intervals
    for the  direct effect and  indirect effect for an IPS
    intervention of $\delta = 2$ applied to the data set from the
    \texttt{mma R} package.}
  \label{table:nde_nie_applic}
\end{table}

From the estimates in Table~\ref{table:nde_nie_applic}, the conclusion may
be easily drawn that there is little total effect of doubling the odds of
participation in a sports team on the BMI of children, based on the data
collected in the observational study made available in the \texttt{mma R}
package. For reference, the marginal odds of participating in a sports team in
the observed data are $0.69$, whereas the odds under the intervention considered
are $1.38$. Based on the 95\% confidence intervals around the point estimates,
we cannot conclude that the proposed incremental propensity score intervention
is sufficiently efficacious to decrease children's BMI. However, the
magnitude of the effects seem to be in the correct direction, with
increased participation in a sports team causing a reduction of BMI of
$0.157$ through changes in behaviors such as snacking and
exercise. Using an approach similar to that demonstrated with this
data set, the direct and indirect effects attributable to
interventions with higher odds of participating in a sports team are
easily estimable.

\section{Discussion}
We have proposed a novel mediation analysis based on the decomposition of the
causal effect of a stochastic intervention on the population, focusing on two
types of interventions: modified treatment policies and exponential tilting.
Unlike the natural direct effect of \cite{Pearl01}, identification of the
(in)direct effect proposed here does not require cross-world counterfactual
independencies, and is therefore achievable in an experimental setting
randomizing both the exposure and the mediator.

We present results for stochastic interventions defined as a modified treatment
policy and explicitly defined in terms of the post-intervention distribution
(exponential tilting). In addition to the considerations about robustness and
smoothness discussed in the present article, the choice between these two
options may also be guided by the fact that modified treatment policies are more
useful in practical settings as they can be used to inform feasible
interventions.

Note that our effect decomposition and estimators allow for
multivariate mediators. The interpretation of the (joint) indirect
effect in this case is entirely context-dependent. For example, the
multivariate mediators may represent an innately multivariate
construct (e.g., a psychological construct such as personality,
behavior, etc.)  in this case the indirect effect could be interpreted
as an effect through the construct (e.g., personality). Nonetheless,
our approach does not require that the multivariate mediators are part
of a single construct; the interpretation in these cases requires more
care.

We assume that there is no mediator-outcome confounder affected by exposure.
Point identification of natural (in)direct effects in the presence of such
variables is not generally possible, and its partial identification is an area
of active research
\citep{robins2010alternative,tchetgen2014bounds,miles2015partial}.

Lastly, for simplicity we focus on estimators constructed as solutions to the
efficient influence function estimating equation; moreover, we have made
implementations of each of the proposed estimators available in the free and
open source \texttt{medshift} software package \citep{hejazi2019medshift} for
the \texttt{R} language and environment for statistical computing \citep{R}.
Alternative estimation strategies, such as targeted minimum loss-based
estimation \citep{vdl2006targeted,vanderLaanRose11,vanderLaanRose18}, may have
better performance than our propsoed estimators in finite samples. The
development of such estimators will be the subject of future research and
software development.

\section{Proofs of results in the main document}
\subsection{Theorem \ref{theo:iden}}
\begin{proof}
  We prove the result separately for exponential tilting and for
  modified treatment policies. First, let $A_{\delta}$ denote a
  variable drawn from the exponentially tilted distribution
  $g_{\delta}(a\mid w)$. We have
  \begin{eqnarray*}
    \E\{Y(A_\delta, Z)\mid  A_\delta=a, Z=z,
    W=w\} & = & \E\{Y(a,z)\mid A_\delta=a, Z=z, W=w\}\\
          & = & \E\{Y(a,z)\mid Z=z, W=w\}\\
          & = & \E\{Y(a,z)\mid A=a, Z=z, W=w\}\\
          & = & m(z,a,w).
  \end{eqnarray*}
  The first equality follows from the definition of
  $Y(A_{\delta}, Z)$, the second equality follows because, by
  definition, $A_{\delta}\indep Y(a,z)\mid (W, Z)$, and the
  third equality follows from \ref{cond:cet}. Finally, the fourth
  equality follows from the consistency implied by he NPSEM:
  $(A,Z)=(a,z)\rightarrow Y(a,z)=Y$.

  Note that, by definition, $A_{\delta}\indep Z\mid W$. Thus
  \[    \E\{Y(A_\delta, Z)\} =\int_{\supp(g_{\delta})\times \supp(q)\times\supp(p)} m(z,a,w) r(z\mid w)
    g_{\delta}(a\mid w)p(w)d\nu(a,z,w),
  \]
  where \ref{cond:cs} ensures that $m(z,a,w)$ is defined in the
  integration set.

  If the intervention is a modified treatment policies, such as our
  Example 1 where $A_{\delta}= d(A, W)$, then the proof proceeds as
  follows. First of all, we have
  $A\indep Y(A_\delta, Z)\mid (A_{\delta}, Z, W)$, so that
  \[\E\{Y(A_\delta, Z)\mid  A_{\delta}=a, A=a', Z=z,
    W=w\} = m(z,a,w).\] Integrating the above expression with respect to
  the joint density of $(A_\delta, A,Z,W)$, and using
  \[r(z\mid w)=\int_{\supp(g)}q(z\mid a', w)g(a'\mid
    w)d\nu(a')\]
  yields the desired result.
\end{proof}

\subsection{DAGs compatible with \ref{cond:cet}}

\begin{lemma}
  Define the non-parametric structural equation model in
  (\ref{eq:npsem}). If ($U_A\indep U_Y$ and $U_W\indep U_Z$ and
  $U_Y\indep U_Z$) and either $U_Y\indep U_W$ or $U_A\indep U_W$, then
  \ref{cond:cet} holds.
\end{lemma}

\begin{proof}
  For fixed $(a,z)$, let $Y(a,z)=f_Y(W, a, z, U_Y)$. $(W,Z)$ d-separates
  $Y(a,z)$ from $A$ in Figures~\ref{fig:dag3} and \ref{fig:dag4},
  concluding the proof of the lemma.
  \begin{figure}[!htb]
    \centering
    \begin{tikzpicture}
      \Vertex{0, 1}{U_W}
      \Vertex{0, 0}{W}
      \Vertex{0, -1}{Z}
      \Vertex{0, -2}{U_Z}
      \Vertex{-1.41, -1}{A}
      \Vertex{-2, -2}{U_A}
      \Vertex{1.41, -1}{Y_{a,z}}
      \Vertex{2, -2}{U_Y}
      \Arrow{W}{A}{black}
      \Arrow{U_A}{A}{black}
      \Arrow{U_Y}{Y_{a,z}}{black}
      \Arrow{U_W}{W}{black}
      \Arrow{U_Z}{Z}{black}
      \Arrow{A}{Z}{black}
      \Arrow{W}{Y_{a,z}}{black}
      \Arrow{W}{Z}{black}
      \EdgeR{U_Y}{U_W}{black}
      \EdgeR{U_A}{U_Z}{black}
    \end{tikzpicture}
    \caption{Directed Acyclic Graph for $U_A \indep U_Y$ and $U_W \indep U_Z$ and
      $U_Y \indep U_Z$ and $U_A \indep U_W$.}
    \label{fig:dag3}
  \end{figure}
  \begin{figure}[!htb]
    \centering
    \begin{tikzpicture}
      \Vertex{0, 1}{U_W}
      \Vertex{0, 0}{W}
      \Vertex{0, -1}{Z}
      \Vertex{0, -2}{U_Z}
      \Vertex{-1.41, -1}{A}
      \Vertex{-2, -2}{U_A}
      \Vertex{1.41, -1}{Y_{a,z}}
      \Vertex{2, -2}{U_Y}
      \Arrow{W}{A}{black}
      \Arrow{U_A}{A}{black}
      \Arrow{U_Y}{Y_{a,z}}{black}
      \Arrow{U_W}{W}{black}
      \Arrow{U_Z}{Z}{black}
      \Arrow{A}{Z}{black}
      \Arrow{W}{Y_{a,z}}{black}
      \Arrow{W}{Z}{black}
      \EdgeL{U_A}{U_W}{black}
      \EdgeR{U_A}{U_Z}{black}
    \end{tikzpicture}
    \caption{Directed Acyclic Graph for $U_A\indep U_Y$ and $U_W\indep U_Z$ and
      $U_Y\indep U_Z$ and $U_Y\indep U_W$.}
    \label{fig:dag4}
  \end{figure}
\end{proof}
\subsection{Theorem~\ref{theo:eif} and Lemmas \ref{lemma:mtp} and
  \ref{lemma:tilt}}\label{sec:proofeif}
\begin{proof}
  In this proof we will use $\Theta(\P)$ to denote a parameter as a
  functional that maps the distribution $\P$ in the model to a
  real number. We will assume that the measure $v$ is discrete so that
  integrals can be written as sums. The resulting influence function
  will also correspond to the influence function of a general measure
  $\nu$. For example, the true parameter value is given by
  \[\theta(\delta)=\Theta(\P) = \sum_{y,z,a,w} m(a, z,
    w)g_{\delta}(a\mid w)p(z,w).\]
  Assume that $g_{\delta}$ is known.  Then the non-parametric MLE of
  $\theta(\delta)$ is given by
  \begin{align}
    \Theta(\Pn)&=\sum_{y,z,a,w}y\Pn(y|a,z,w)g_{\delta}(a\mid w)\Pn(z,w)\notag\\
               &=\sum_{y,z,a,w}y\frac{\Pn f_{y,a,z,w}}{\Pn
                 f_{a,z,w}}g_{\delta}(a\mid w)\Pn f_{z,w}\label{nonpest},
  \end{align}
  where we remind the reader of the notation $Pf =\int f dP$. Here
  $f_{y,a,z,w}=I(Y=y,A=a,Z=z,W=w)$, $f_{a,z,w}=I(A=a,Z=z,W=w)$ ,
  $f_{a,w}=I(A=a,W=w)$, $f_{z,w}=I(Z=z,W=w)$, and $I(\cdot)$ denotes
  the indicator function.

  We will use the fact that the efficient influence function in a
  non-parametric model corresponds with the influence curve of the
  NPMLE. This is true because the influence curve of any regular
  estimator is also a gradient, and a non-parametric model has only
  one gradient. Appendix 18 of \cite{vanderLaanRose11} shows that if
  $\hat \Theta(\Pn)$ is a substitution estimator such that
  $\theta(\delta)=\hat \Theta(\P)$, and $\hat \Theta(\Pn)$ can be written as
  $\hat \Theta^*(\Pn f:f\in\mathcal{F})$ for some class of functions
  $\mathcal{F}$ and some mapping $B^*$, the influence curve of
  $\hat \Theta(\Pn)$ is equal to
  \[IC(\P)(O)=\sum_{f\in\mathcal{F}}\frac{d\hat \Theta^*(\P)}{d\P f}\{f(O)-\P f\}.\]

  Applying this result to (\ref{nonpest}) with
  $\mathcal{F}=\{f_{y,a,z,w},f_{y,a,w},f_{a,w},f_{z,w},f_w\}$ gives an
  efficient influence function equal to
  $D^Y_{\eta,\delta}(o) + D^{Z,W}_{\eta,\delta}(o) -
  \theta(\delta)$. It remains to find the component $D^A_{\eta,\delta}(o)$ for
  each specific intervention. This component may be found as the IF
  of the estimator
  \[\Theta(\Pn)=\sum_{y,z,a,w}y\P(y|a,z,w)\hat g_{\delta}(a\mid
    w)\P(z,w),\]
  where $\hat g_{\delta}$ is the MLE of $g_{\delta}$, obtained by
  substitution of the MLE of $g$.

  The algebraic derivations described here are lengthy and not
  particularly illuminating, and are therefore omitted from the proof.
\end{proof}

\subsection{Proof of Corollary \ref{coro:tilt}}
\begin{proof}
  In this proof we use the notation $\pi(w)=g(1\mid w)$. From the
  parameterization
  $e(a\mid z, w) = g(a\mid w)q(z\mid a, w) / r(z\mid w)$, note
  that
  \[\phi(a,w) = \int
    m(a,z,w)\dd \P(z\mid w).\] Thus, $D^A_{\eta,\delta}(o)$ from
  Lemma~\ref{lemma:tilt} may be written as follows
  \begin{align*}
    D^A_{\eta,\delta}(o) &= \sum_{t\in\{0,1\}}\frac{g_{\delta}(a\mid w)}{g(a\mid w)}\int
                           m(t,z,w)\dd \P(z\mid w)\{I(a=t) - g_\delta(t\mid w)\}\\
    =& \sum_{t\in\{0,1\}}\int
       m(t,z,w)\dd \P(z\mid w)\left[a\frac{\pi_\delta(w)}{\pi(w)}\{t - g_\delta(t\mid w)\}+\right.\\
                         &\left.(1-a)\frac{1-\pi_\delta(w)}{1-\pi(w)}\{1-t - g_\delta(t\mid w)\}\right]
  \end{align*}
  Note that
  \[\pi_\delta(w)\{t - g_\delta(t\mid w)\} =
    -\{1-\pi_\delta(w)\}\{1-t - g_\delta(t\mid w)\} =
    (2t-1)\pi_\delta(w)\{1-\pi_\delta(w)\}.\]
  Thus,
  \begin{align*}
    D^A_{\eta,\delta}(o)  =& \sum_{t\in\{0,1\}}\int
                             m(t,z,w)\dd \P(z\mid w) (2t-1)\pi_\delta(w)\{1-\pi_\delta(w)\}\left[\frac{a}{\pi(w)}-
                             \frac{1-a}{1-\pi(w)}\right]\\
    =&\frac{\pi_\delta(w)\{1-\pi_\delta(w)\}}{\pi(w)\{1-\pi(w)\}}\{a - \pi(w)\}\sum_{t\in\{0,1\}}\int
       m(t,z,w)\dd \P(z\mid w)
       (2t-1)
  \end{align*}
  since
  \[\frac{\pi_\delta(w)\{1-\pi_\delta(w)\}}{\pi(w)\{1-\pi(w)\}}
    = \frac{\delta}{\{\delta\pi(w)+1-\pi(w)\}^2},\]
  expanding the sum in $t$ concludes the proof.
\end{proof}
\subsection{Proof of Theorem \ref{ass:lin1}}
\begin{proof}
  Let $\Pnj$ denote
  the empirical distribution of the prediction set ${\cal V}_j$, and let
  $\Gnj$ denote the associated empirical process
  $\sqrt{n/J}(\Pnj-\P)$. Note that
  \[\thetaaipw = \frac{1}{J}\sum_{j=1}^J\Pnj D_{\hat
      \eta_j,\delta},\,\,\,\theta(\delta)=\P D_{\eta}.\]Thus,
  \[  \sqrt{n}\{\thetaaipw - \theta(\delta)\}=\Gn \{D_{\eta,\delta}
    - \theta(\delta)\} + R_{n,1}(\delta) + R_{n,2}(\delta),\]
  where
  \[  R_{n,1}(\delta)  =\frac{1}{\sqrt{J}}\sum_{j=1}^J\Gnj(D_{\hat
      \eta_j,\delta} - D_{\eta,\delta}),\,\,\,
    R_{n,2}(\delta)  = \frac{\sqrt{n}}{J}\sum_{j=1}^J\P\{D_{\hat
      \eta_j,\delta}-\theta(\delta)\}.
  \]
  It remains to show that $R_{n,1}(\delta)$ and $R_{n,2}(\delta)$ are
  $o_P(1)$.  Theorem \ref{theo:dr1} together with the Cauchy-Schwartz
  inequality and assumption \ref{ass:sec1} of the theorem shows that
  $||R_{n,2}||_{\Delta}=o_P(1)$. For $||R_{n,1}||_{\Delta}$ we use
  empirical process theory to argue conditional on the training sample
  ${\cal T}_j$. In particular, Lemma 19.33 of \cite{vanderVaart98}
  applied to the class of functions
  ${\cal F} = \{D_{\hat \eta_j,\delta} - D_{\eta,\delta}\}$ (which
  consists of one element) yields
  \[E\left\{\big|\Gnj (D_{\hat \eta_j,\delta} - D_{\eta,\delta})\big|
      \,\bigg|\, {\cal T}_j\right\}\lesssim \frac{2C\log 2}{n^{1/2}} +
    ||D_{\hat \eta_j,\delta} - D_{\eta,\delta}||(\log 2)^{1/2}\] By
  assumption \ref{ass:sec1}, the left hand side is $o_P(1)$. Lemma 6.1
  of \cite{chernozhukov2018double} may now be used to argue that
  conditional convergence implies unconditional convergence, concluding
  the proof.

\end{proof}

\subsection{Proof of Theorem \ref{ass:lin2}}
Let $||f||_{\Delta}=\sup_{\delta\in\Delta}|f(\delta)|$. Let $\Pnj$
denote the empirical distribution of the prediction set ${\cal V}_j$,
and let $\Gnj$ denote the associated empirical process
$\sqrt{n/J}(\Pnj-\P)$. Note that
\[\thetaaipw = \frac{1}{J}\sum_{j=1}^J\Pnj D_{\hat
    \eta_j,\delta},\,\,\,\theta(\delta)=\P D_{\eta}.\]Thus,
\[  \sqrt{n}\{\thetaaipw - \theta(\delta)\}=\Gn \{D_{\eta,\delta}
  - \theta(\delta)\} + R_{n,1}(\delta) + R_{n,2}(\delta),\]
where
\[  R_{n,1}(\delta)  =\frac{1}{\sqrt{J}}\sum_{j=1}^J\Gnj(D_{\hat
    \eta_j,\delta} - D_{\eta,\delta}),\,\,\,
  R_{n,2}(\delta)  = \frac{\sqrt{n}}{J}\sum_{j=1}^J\P\{D_{\hat
    \eta_j,\delta}-\theta(\delta)\}.
\]
The map $\delta\mapsto \bar D_{\eta,\delta}$ is Lipschitz, which
implies that the class
${\cal F} = \{\bar D_{\eta,\delta}:\delta \in \Delta\}$ has bounded
bracketing numbers \citep[Theorem 2.7.11
of][]{vanderVaart&Wellner96}. Therefore, $\cal F$ is Donsker and
$\Gn \{D_{\eta,\delta} - \theta(\delta)\} \rightsquigarrow \mathbb
G(\delta)$ in $\ell^\infty(\Delta)$.

It remains to show that $||R_{n,1}||_{\Delta}$ and
$||R_{n,2}||_{\Delta}$ are $o_P(1)$. Theorem \ref{theo:dr2} together
with the assumptions of the theorem and the Cauchy-Schwartz
inequality, show that $||R_{n,2}||_{\Delta}=o_P(1)$. For
$||R_{n,1}||_{\Delta}$ we use empirical process theory to argue
conditional on the training sample ${\cal T}_j$. Let
${\cal F}_n^j=\{D_{\hat \eta_j,\delta} -
D_{\eta,\delta}:\delta\in\Delta\}$. Because the function
$\hat\eta_j$ is fixed given the training data, we can apply Theorem
2.14.2 of \cite{vanderVaart&Wellner96} to obtain
\[E\left\{\sup_{f\in {\cal F}_n^j}|\Gnj f| \,\,\bigg|\,\, {\cal
      T}_j\right\}\lesssim ||F^j_n||\int_0^1\sqrt{1+N_{[\,]}(\epsilon
    ||F_n^j||, {\cal F}_n^j, L_2(\P))}\dd\epsilon, \] where
$N_{[\,]}(\epsilon ||F_n^j||, {\cal F}_n^j, L_2(\P))$ is the
bracketing number and we take
$F_n^j=\sup_{\delta\in\Delta}|D_{\hat \eta_j,\delta} -
D_{\eta,\delta}|$ as an envelope for the class ${\cal
  F}_n^j$. Theorem 2.7.2 of \cite{vanderVaart&Wellner96} shows
\[\log N_{[\,]}(\epsilon ||F_n^j||, {\cal F}_n^j, L_2(\P))\lesssim
  \frac{1}{\epsilon ||F_n^j||}.\]
This shows
\begin{align*}||F^j_n||\int_0^1\sqrt{1+N_{[\,]}(\epsilon
  ||F_n^j||, {\cal F}_n^j, L_2(\P))}\dd\epsilon &\lesssim
                                                  \int_0^1\sqrt{||F^j_n||^2+\frac{||F^j_n||}{\epsilon}}\dd\epsilon\\
                                                &\leq
                                                  ||F^j_n||+||F^j_n||^{1/2}\int_0^1\frac{1}{\epsilon^{1/2}}\dd\epsilon\\
                                                &\leq ||F^j_n|| + 2
                                                  ||F^j_n||^{1/2}.
\end{align*}
Since $||F^j_n||=o_P(1)$, this shows
$\sup_{f\in {\cal F}_n^j}\Gnj f=o_P(1)$ for each $j$, conditional on
${\cal T}_j$.  and thus $||R_{n,1}||_{\Delta}=o_P(1)$, concluding the
proof of the theorem.

\section{Second order representation of the expectation of the EIF}
\begin{theorem}\label{theo:dr1}
  Let $d(A,W)$ satisfy assumption \ref{ass:inv}. Denote
  $m_d(z,a,w)=m(z,d(a,w), w)$. Let $q_1(z\mid a, w)$ denote any density
  compatible with $\phi_1(a,w)$. That is, let $q_1$ be such that
  \[\phi_1(a,w) = \int \frac{g_1(a\mid w)}{e_1(a\mid z, w)}m_{1,d}(z,a,w)q_1(z\mid
    a,w)d\nu(z),\] and define $r_1 = g_1q_1/e_1$. In this theorem we
  denote $\dd\xi(w) =\dd \P(w)$. We have
  \begin{eqnarray*}
    \P D_{\eta_1, \delta} - \theta(\delta)&=&\int g_{\delta}\left(\frac{e}{e_1} -
                                              1\right)(m-m_1)r\dd\kappa\dd\xi\\ &&+ \int\frac{e}{e_1}(g_{\delta,1} -
                                                                                   g_{\delta})(m-m_1)r\dd\kappa\dd\xi \\&&+\int m_{1,d}(r_1-r)(g-g_1)\dd\kappa\dd\xi
  \end{eqnarray*}
\end{theorem}

\begin{proof}
  Note that
  \begin{eqnarray}
    \P D^Y_{\eta_1, \delta} + \P D^{Z,W}_{\eta,\delta}-\theta(\delta)
    &=&\int\frac{g_{\delta}}{e_1}(m-m_1)er\dd\kappa\dd\xi - \int
        g_{\delta}(m-m_1)r\dd\kappa\dd\xi\notag\\ &&+\int
                                                     \frac{e}{e_1}(g_{1,\delta}-g_{\delta})(m-m_1)r\dd\kappa\dd\xi
                                                     + \int m(g_{1,\delta} -
                                                     g_{\delta})r\dd\kappa\dd\xi\notag\\
    &=&\int g_{\delta}\left(\frac{e}{e_1}-1\right)(m-m_1)r\dd\kappa\dd\xi\label{eq:valid}\\
    &&+\int\frac{e}{e_1}(g_{1,\delta} -
       g_{\delta})(m-m_1)r\dd\kappa\dd\xi\notag\\
    &&+\int(g_{1,\delta} - g_{\delta})m_1r\dd\kappa\dd\xi.\notag
  \end{eqnarray}
  We have
  \begin{equation}
    \P D^A_{\eta_1,\delta}=  \P\left(\phi_1 - \int\phi_1 g_1d\nu\right) =\int (g - g_1)m_{1,d}\frac{g_1q_1}{e_1}\dd\kappa\dd\xi.\label{eq:au1}
  \end{equation}
  Under \ref{ass:inv}, we can change variables in the following integral to obtain
  \begin{align*}
    \int(g_{1,\delta} - g_{\delta})m_1r\dd\kappa\dd\xi & =\int(g_1 -
                                                         g)m_{1,d}r\dd\kappa\dd\xi\\
                                                       & =\int(g_1 -
                                                         g)m_{1,d}\frac{gq}{e}\dd\kappa\dd\xi,
  \end{align*}
  where we used the fact that
  $r(z\mid w) = g(a\mid w)q(z\mid a, w)/e(a\mid z,w)$. Adding
  this quantity in both sides of (\ref{eq:au1}) we get
  \begin{eqnarray*}
    \P D^A_{\eta_1,\delta}+\int(g_1 - g)m_{1,d}r\dd\kappa\dd\xi &=&
                                                                    \int\frac{e_1q_1}{g_1}m_{1,d}(g_1-g)\dd\kappa\dd\xi\\ &&-
                                                                                                                             \int\frac{eq}{g}m_{d}(g-g_1)\dd\kappa\dd\xi\\
                                                                &&+
                                                                   \int\frac{eq}{g}(m_{d}-m_{1,d})(g-g_1)\dd\kappa\dd\xi\\
                                                                &=& \int
                                                                    (\phi_1-\phi)(g-g_1)\dd\kappa\dd\xi\\
                                                                &&+ \int\frac{eq}{g}(m_{d}-m_{1,d})(g-g_1)\dd\kappa\dd\xi.
  \end{eqnarray*}
\end{proof}

\begin{theorem}\label{theo:dr2}
  Define $c(w) = \{\int_a \exp(\delta a)g(a\mid w)\}^{-1}$, and let
  $c_1(w)$ be defined analogously. Let $b(a) = \exp(\delta a)$. Using
  the same notation as in Theorem~\ref{theo:dr1}, we have
  \begin{eqnarray*}
    \P D_{\eta_1, \delta} - \theta(\delta)&=&\int g_{\delta}\left(\frac{e}{e_1} -
                                              1\right)(m-m_1)r\dd\kappa\dd\xi\\
                                          &&+ \int\frac{e}{e_1}(g_{1,\delta} -
                                             g_{\delta})(m-m_1)r\dd\kappa\dd\xi\\
                                          && +\int(g_{1,\delta} -
                                             g_{\delta})\{(m_1-m)r-(\phi_1-\phi)\}\dd\kappa\dd\xi\\
                                          &&-\int\left\{(c_1-c)^2\int b
                                             g_1\phi\dd\kappa\int b
                                             g\dd\kappa\right\}\dd\xi\\
                                          &&+\int \left\{(c_1-c)\int b\phi(g-g_1)\dd\kappa\right\}\dd\xi
  \end{eqnarray*}

\end{theorem}

\begin{proof}
  Let
  $q_1(z\mid a,w)$ be any density compatible with $\phi_1$. That is
  \[\phi_1(a,w)=\int \frac{g_1(a\mid w)}{e_1(a\mid z,w)}m_1(a,z,w)q_1(z\mid a,w)\dd\nu(z)\]
  Note that display (\ref{eq:valid}) is also valid here. Note also that
  \[\int(g_{1,\delta} - g_{\delta})m_1r\dd\kappa\dd\xi = \int(g_{1,\delta} -
    g_{\delta})\phi\dd\kappa\dd\xi + \int(g_{1,\delta} -
    g_{\delta})\{(m_1-m)r-(\phi_1-\phi)\}\dd\kappa\dd\xi.\] Note that
  $g_{1,\delta}(a\mid w) = c_1(w)b(a)g_1(a\mid w)$. We have
  \begin{align}
    \P D_{\eta_1}^A +& \int(g_{\delta,1} - g_{\delta})\phi\dd\xi\notag\\
                     &=\int\left\{\int\frac{g_{1,\delta}}{g_1}\phi g\dd\kappa
                       -\int \frac{g_{1,\delta}}{g_1}g\dd\kappa\int\phi g_{1,\delta}\dd\kappa + \int(g_{1,\delta} -
                       g_\delta)\phi\dd\kappa\right\}\dd\xi\notag\\
                     &=\int\left\{\frac{g_{1,\delta}}{g_1}g\phi\dd\kappa - \int g_\delta\phi\dd\kappa + \int
                       g_{1,\delta}\phi\dd\kappa\left[1-\int\frac{g_{1,\delta}}{g_1}g\dd\kappa\right]\right\}\dd\xi\notag\\
                     &=\int\left\{c_1\int
                       b\phi g\dd\kappa
                       - c_1\int
                       b\phi
                       g\dd\kappa
                       + c_1\int
                       bg\phi\dd\kappa\int(c-c_1)bg\dd\kappa\right\}\dd\xi\notag\\ &=\int(c_1-c)\left\{\int b\phi g\dd\kappa - c_1\int b
                                                                                     g_1\phi\dd\kappa\int b g\dd\kappa\right\}\dd\xi\notag\\
                     &=\int(c_1-c)\left\{\int b\phi g\dd\kappa - c\int b
                       g_1\phi\dd\kappa\int b g\dd\kappa - (c_1-c)\int b
                       g_1\phi\dd\kappa\int b g\dd\kappa\right\}\dd\xi\notag\\
                     &=\int\left\{-(c_1-c)^2\int b
                       g_1\phi\dd\kappa\int b g\dd\kappa + (c_1-c)\left[\int
                       b\phi g\dd\kappa - \int
                       bg_1\phi\dd\kappa\right]\right\}\dd\xi\label{eq:intone}\\
                     &=\int\left\{-(c_1-c)^2\int b
                       g_1\phi\dd\kappa\int b g\dd\kappa + (c_1-c)\int b\phi(g-g_1)\dd\kappa\right\}\dd\xi\notag,
  \end{align}
  where (\ref{eq:intone}) follows from $c\int b g\dd\kappa = 1$
\end{proof}

\bibliographystyle{plainnat}
\bibliography{refs}

\end{document}